\documentclass[11pt]{elsarticle}
\usepackage{enumerate}
\usepackage{hyperref}
\usepackage{amsmath,amssymb,amsthm}
\usepackage{tikz}
\usepackage{subcaption}

\usepackage{algorithm}%
\usepackage{algpseudocode}
\usepackage[margin=1.3in]{geometry}
\usepackage{todonotes}

\hypersetup{
    pdftitle = {A Single-Exponential Fixed-Parameter Tractable Algorithm for Distance-Hereditary Vertex Deletion},
   pdfauthor=  {Eduard Eiben and Robert Ganian and O-joung Kwon}
   }

\newcommand\abs[1]{\lvert #1\rvert}
\newtheorem{theorem}{Theorem}[section]
\newtheorem{lemma}[theorem]{Lemma}

\newtheorem{PROP}[theorem]{Proposition}
\newtheorem{observation}[theorem]{Observation}

\newtheorem{claim}{Claim}
\newenvironment{clproof}{\begin{list}{}{%
              \setlength{\leftmargin}{5mm}%
              } \item {\it Proof.} }{\hfill$\lozenge$\end{list}\medskip}
              
\newtheorem{definition}{Definition}
\newtheorem{RRULE}{Reduction Rule}
\newtheorem{BRULE}{Branching Rule}
\newcommand{\YES}{\textsc{Yes}}
\newcommand{\NO}{\textsc{No}}

\newcommand{\cc}{\operatorname{cc}}
\newcommand{\twinclass}{\textbf{tc}(G-S)}
\newcommand{\DHVD}{\textsc{Distance-Hereditary Vertex Deletion}}
\newcommand{\disjointDHVD}{\textsc{Disjoint Distance-Hereditary Vertex Deletion}}

\newcommand{\comp}{\operatorname{comp}}
\newcommand{\btwn}{\operatorname{path}}
\newcommand{\bigoh}{\mathcal{O}}

\begin{document}
\title{A Single-Exponential Fixed-Parameter Algorithm for Distance-Hereditary Vertex Deletion}
\author[Vienna]{Eduard Eiben}
\ead{eduard.eiben@gmail.com}
\author[Vienna]{Robert Ganian}
\ead{rganian@gmail.com}
\author[Kwon]{O-joung Kwon}
\ead{ojoungkwon@gmail.com}
\address[Vienna]{Algorithms and Complexity Group, TU Wien, Vienna, Austria}
\address[Kwon]{Logic and Semantics, Technische Universit\"at Berlin, Germany}
\date{\today}
\begin{abstract}
Vertex deletion problems ask whether it is possible to delete at most $k$ vertices from a graph so that the resulting graph belongs to a specified graph class. Over the past years, the parameterized complexity of vertex deletion to a plethora of graph classes has been systematically researched. Here we present the first single-exponential fixed-parameter tractable algorithm for vertex deletion to distance-hereditary graphs, a well-studied graph class which is particularly important in the context of vertex deletion due to its connection to the graph parameter rank-width. We complement our result with matching asymptotic lower bounds based on the exponential time hypothesis. As an application of our algorithm, we show that a vertex deletion set to distance-hereditary graphs can be used as a parameter which allows single-exponential fixed-parameter tractable algorithms for classical NP-hard problems.
\end{abstract}

\maketitle

\section{Introduction}\label{sec:introduction}
Vertex deletion problems include some of the best studied NP-hard problems in theoretical computer science, including \textsc{Vertex Cover} or \textsc{Feedback Vertex Set}. In general, the problem asks whether it is possible to delete at most $k$ vertices from a graph so that the resulting graph belongs to a specified graph class. While these problems are studied in a variety of contexts, they are of special importance for the parameterized complexity paradigm~\cite{DowneyF13,CyganFKLMPPS15}, which measures the performance of algorithms not only with respect to the input size but also with respect to an additional numerical parameter. The notion of vertex deletion allows a highly natural choice of the parameter (specifically, $k$), especially for problems where the solution size is not defined or cannot be used. Many vertex deletion problems are known to admit so-called \emph{single-exponential fixed-parameter tractable (FPT) algorithms}, which are algorithms running in time $\bigoh(c^k\cdot n^{\bigoh(1)})$ for input size $n$ and some constant $c$.

Over the past years, the parameterized complexity of vertex deletion to a plethora of graph classes has been systematically researched, 
and in particular, if the target class admits efficient algorithms for many NP-hard problems, then such a class get more attention.
For this reason, classes of graphs of constant treewidth have been studied in detail, 
and Fomin et al.~\cite{FominLMS12} and Kim et al.~\cite{KLPRRSS13} showed that the corresponding \textsc{Treewidth-$t$ Vertex Deletion}\footnote{\textsc{Treewidth-$t$ Vertex Deletion} asks whether it is possible to delete $k$ vertices so that the resulting graph has treewidth at most $t$.}
problem is solvable in single-exponential FPT time.
Interestingly, this problem is a special case of general \textsc{Planar $\mathcal{F}$-Deletion} problems, which ask whether one can hit all of minor models of graphs in $\mathcal{F}$ by at most $k$ vertices, when $\mathcal{F}$ contains at least one planar graph. 
The condition that $\mathcal{F}$ contains a planar graph is essential because it tells that the outside of any solution should have bounded treewidth, by the grid-minor theorem~\cite{RobertsonS1986}.
Several authors~\cite{FominLMS12,KLPRRSS13} have used this fact to design single-exponential FPT algorithms.

The successful development of single-exponential FPT algorithms for \textsc{Treewidth-$t$ Vertex Deletion} motivates us to study \textsc{Rank-width $t$-Vertex Deletion}, which is analogous to \textsc{Treewidth-$t$ Vertex Deletion} but replaces treewidth with rank-width.
Rank-width~\cite{OS2004, Oum05} is a graph parameter introduced for generalizing graph classes of bounded treewidth into dense graph classes; for example,
complete graphs have unbounded treewidth but rank-width $1$.
Generally, classes of graphs of bounded rank-width capture the graphs that can be recursively decomposable along vertex bipartitions $(A,B)$ where the number of distinct neighborhood types from one side to the other is bounded. 
Courcelle, Makowski, and Rotics~\cite{CourcelleMR2000} proved that every MSO$_1$-expressible problem can be solved in
polynomial time on graphs of bounded rank-width (see also the work of Ganian and Hlin\v en\' y~\cite{GanianH10}).

Kant\'e et al.~\cite{KanteKKP2015} observed that \textsc{Rank-width $t$-Vertex Deletion} is fixed parameter tractable using the general framework of Courcelle, Makowski, and Rotics.
However, this algorithm does not provide any reasonable function for $k$. Thus Kant\'e et al. naturally asked whether it is solvable in reasonably better running time.
For instance, it is actually open whether \textsc{Rank-width $t$-Vertex Deletion} can even be solved in time $2^{2^{\mathcal{O}(k)}}n^{\mathcal{O}(1)}$, where $k$ is the size of the deletion set.

In this paper, we focus on graphs of rank-width at most $1$, which are \emph{distance-hereditary graphs}.
Distance-hereditary graphs were introduced by Howorka~\cite{howorka77} in 1977, long before the discovery of rank-width~\cite{OS2004} and the observation by Oum~\cite{Oum05} that the class of graphs of rank-width at most $1$ are precisely distance-hereditary graphs.
Bandelt and Mulder~\cite{BM1986} found all the minimal induced subgraph obstructions for distance-hereditary graphs. 
Distance-hereditary graphs are naturally related to split decompositions, where they are exactly the graphs that are completely decomposable into stars and complete graphs~\cite{Bouchet1988a}.
We explain these structural properties in more detail in Section~\ref{sec:prelim}.
This structure has led to the development of a number of algorithms for distance-hereditary graphs~\cite{CogisT2005, HungC2005, HsiehHH2006, RuoM2007, MarcD2007, NakanoUU2007, GassnerH2008}.
Given the above, we view the vertex deletion problem for distance-hereditary graphs as a first step towards handling \textsc{Rank-width-$t$ Vertex Deletion}.

\subsection*{Our Contribution.}
A graph $G$ is called \emph{distance-hereditary} if for every connected induced subgraph $H$ of $G$ and every $v,w\in V(H)$, 
the distance between $v$ and $w$ in $H$ is the same as the distance between $v$ and $w$ in $G$. 
We study the following problem.

\medskip
\noindent
\fbox{\parbox{0.97\textwidth}{
\DHVD \\
\emph{Instance:} A graph $G$ and an integer $k$. \\
\emph{Parameter:} $k$. \\
\emph{Task:} Is there a vertex set $Q\subseteq V(G)$ with $|Q|\le k$ such that $G-Q$ is distance-hereditary? }}
\vskip 0.2cm

The main result of this paper is a single-exponential FPT algorithm for \DHVD.

\begin{theorem}
\DHVD\ can be solved in time $\bigoh(37^k\cdot |V(G)|^7(|V(G)|+|E(G)|))$.
\end{theorem}

We note that this solves an open problem of Kant\' e, Kim, Kwon, and Paul~\cite{KanteKKP2015}. The core of our approach exploits two distinct characterizations of distance-hereditary graphs: one by forbidden induced subgraphs (obstructions), and the other by admitting a special kind of split decomposition~\cite{Cunningham1982}.

The algorithm can be conceptually divided into three parts. 

\begin{enumerate}
\item \textbf{Iterative Compression}. This technique allows us to reduce the problem to the easier \textsc{Disjoint Distance-Hereditary Vertex Deletion}, where we assume that the instance additionally contains a certain form of advice to aid us in our computation. Specifically, this advice is a vertex deletion set $S$ to distance-hereditary graphs which is disjoint from and slightly larger than the desired solution.
  
  \item \textbf{Branching Rules}. We exhaustively apply two branching rules to simplify the given instance of \textsc{Disjoint Distance-Hereditary Vertex Deletion}. At a high level, these branching rules allow us to assume that the resulting instance contains no small obstructions and furthermore that certain connectivity conditions hold on $G[S]$.
  
  \item \textbf{Simplification of Split Decomposition}. We compute the split decomposition of $G-S$ and exploit the properties of our instance $G$ guaranteed by branching to prune the decomposition. In particular, we show that the connectivity conditions and non-existence of small obstructions mean that $S$ must interact with the split decomposition of $G-S$ in a special way, and this allows us to identify irrelevant vertices in $G-S$. This is by far the most technically challenging part of the algorithm.
\end{enumerate}

A more detailed explanation of our algorithm is provided in Section~\ref{sec:general}, after the definition of required notions. We complement this result with an algorithmic lower bound which rules out a subexponential FPT algorithm for \textsc{Distance-Hereditary Vertex Deletion} under well-established complexity assumptions. We also note that the naive approach of simply hitting all known ``obstructions'' (i.e., forbidden induced subgraphs) for distance-hereditary graphs does not lead to an FPT algorithm. Indeed, the set of induced subgraph obstructions for distance-hereditary graphs includes induced cycles of length at least $5$.
Heggernes et al.~\cite{Heggernes2013} showed that 
the problem asking whether it is possible to delete $k$ vertices so that the resulting graph has no induced cycles of length at least $5$ is W[2]-hard.
Therefore, unless the Exponential Time Hypothesis fails, one cannot obtain a single-exponential FPT algorithm for \DHVD\ by simply finding and hitting all forbidden induced subgraphs for the class.

The paper is organized as follows. Section~\ref{sec:prelim} contains the necessary preliminaries and notions required for our results. In Section~\ref{sec:general}, we set the stage for the process of simplifying the split decomposition, which entails the definition of \textsc{Disjoint Distance-Hereditary Vertex Deletion}, introduction of our branching rules, and a few technical lemmas which will be useful throughout the later sections.
Section~\ref{sec:rules} then introduces and proves the safeness of five polynomial-time reduction rules; crucially, the exhaustive application of these rules guarantees that the resulting instance will have a certain ``inseparability'' property. Using this structural result, we prove that one of reduction rules is applicable until the remaining instance is trivial. Finally, the proof of our main result as well as the corresponding lower bound are presented in Section~\ref{sec:completing}. Section~\ref{sec:completing} also illustrates one potential application of our result: we show that a vertex deletion set to distance-hereditary graphs can be used as a parameter which allows single-exponential FPT algorithms for classical NP-hard problems.

\section{Preliminaries}\label{sec:prelim}

All graphs in this paper are simple and undirected.
For a graph $G$, let $V(G)$ and $E(G)$ denote the vertex set and the edge set of $G$, respectively. 
For $S\subseteq V(G)$, let $G[S]$ denote the subgraph of $G$ induced by $S$. For $v\in V(G)$ and $S\subseteq V(G)$, let $G- v$ be the graph obtained from $G$ by removing $v$, and let $G-S$ be the graph obtained by removing all vertices in $S$. 
For $F\subseteq E(G)$, let $G-F$ denote the graph obtained from $G$ by removing all edges in $F$.
For $v\in V(G)$, the set of neighbors of $v$ in $G$ is denoted by $N_G(v)$.
For $A\subseteq V(G)$, let $N_G(A)$ denote the set of vertices in $G-A$ that have a neighbor in $A$.
We denote by $\cc(G)$ the number of connected components of $G$.
An edge $e$ of a connected graph $G$ is a \emph{cut edge} if the graph obtained from $G$ by removing $e$ is disconnected.

The \emph{length} of a path is the number of edges on the path.
For $v\in V(G)$ and a subgraph $H$ of $G-v$, we say $v$ is adjacent to $H$ if it has a neighbor in $H$.
A \emph{star} is a tree with a distinguished vertex, called the \emph{center}, adjacent to all other vertices. A \emph{complete graph} is a graph with all possible edges.

Two vertices $v$ and $w$ in a graph $G$ are called \emph{twins} if they have the same set of neighbors in $V(G)\setminus \{v,w\}$. 
For two vertex sets $A$ and $B$, we say that 
\begin{itemize}
\item $A$ is \emph{complete} to $B$ if for every $a\in A$, $b\in B$, $a$ is adjacent to $b$, 
\item $A$ is \emph{anti-complete} to $B$ if for every $a\in A$, $b\in B$, $a$ is not adjacent to $b$.
\end{itemize}

In parameterized complexity, an instance of a parameterized problem consists in a pair $(x,k)$, where $k$ is a secondary measurement, called the \emph{parameter}. 
A parameterized problem $Q\subseteq \Sigma^* \times N$ is \emph{fixed-parameter tractable} (\emph{FPT}) if there is an algorithm which decides whether $(x,k)$ belongs to $Q$ in time $f(k)\cdot \abs{x}^{\mathcal{O}(1)}$ for some computable function $f$.

\subsection{Distance-Hereditary Graphs}
A graph $G$ is called \emph{distance-hereditary} if for every connected induced subgraph $H$ of $G$ and every $v,w\in V(H)$, 
the distance between $v$ and $w$ in $H$ is the same as the distance between $v$ and $w$ in $G$. 
For instance, the induced cycle $c_1c_2c_3c_4c_5c_1$ is not distance-hereditary, because the distance from $c_1$ to $c_3$ is $2$, 
but if we take an induced subgraph on $\{c_1, c_3, c_4, c_5\}$, then the distance becomes $3$.
This graph class was first introduced by Howorka~\cite{howorka77}, and deeply studied by Bandelt and Mulder~\cite{BM1986}.
There are several other, equivalent characterizations of distance-hereditary graphs. One of the most prominent ones links it to the structural parameter \emph{rank-width}~\cite{Oum05}; specifically, distance-hereditary graphs are precisely the graphs of rank-width $1$~\cite{Oum05}.
However, in this paper we will exploit two other characterizations of the graph class: one by forbidden induced subgraphs (given below), and one via split decompositions (given in the following subsection).

The house, the gem, and the domino graphs are depicted in Figure~\ref{fig:obsdh}.
A graph isomorphic to one of the house, the gem, the domino, and induced cycles of length at least $5$ will be called a \emph{distance-hereditary obstruction} or shortly a \emph{DH obstruction}.
A DH obstruction with at most $6$ vertices will be called a \emph{small DH obstruction}.
Note that every DH obstruction does not contain any twins.

\begin{figure}[t]
\centerline{\includegraphics[scale=1]{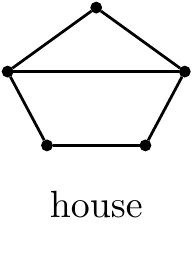} \quad\quad
\includegraphics[scale=1]{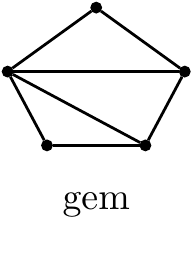} \quad\quad
\includegraphics[scale=1]{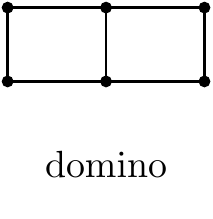} }
\caption{Small DH obstructions which are not cycles.}
\label{fig:obsdh}
\end{figure}

\begin{theorem}[Bandelt and Mulder~\cite{BM1986}]
A graph is distance-hereditary if and only if it contains no DH obstructions as induced subgraphs.
\end{theorem}
We state an observation which will be useful later on.

\begin{observation}
\label{obs:sub}
For any DH obstruction $H$ and any edge $e$ in $H$, it holds that the graph $H'$ obtained by subdividing $e$ also contains a DH obstruction as an induced subgraph.
\end{observation}
The following lemma will be used to find DH obstructions later on.

\begin{lemma}[Kant\`e, Kim, Kwon, and Paul, Lemma 4.3 of \cite{KanteKKP2015}]\label{lem:dhobs}
Let $G$ be a graph obtained from an induced path of length at least $3$ by adding a vertex $v$ adjacent to its end vertices where
$v$ may be adjacent to some internal vertices of the path.
Then $G$ has a DH obstruction containing $v$.
In particular, if the given path has length at most $4$, then $G$ has a small DH obstruction containing $v$.
\end{lemma}
\begin{proof}
The first statement was shown in Lemma 4.3 of \cite{KanteKKP2015}. If the given path has length at most $4$, then $G$ has at most $6$ vertices, and thus $G$ contains a small DH obstruction containing $v$.
\end{proof}

\subsection{Split decompositions}

We follow the notations used by Bouchet~\cite{Bouchet1988a}.
A \emph{split} of a connected graph $G$ is a vertex partition $(X,Y)$ of $G$ such that $\abs{X}\ge 2, \abs{Y}\ge 2$, and $N_G(Y)$ is complete to $N_G(X)$. 
See Figure~\ref{fig:examplesplit} for an example.
Splits are also called \emph{1-joins}, or simply \emph{joins}~\cite{GSH1989}. 
A connected graph $G$ is called a \emph{prime graph} if $\abs{V(G)}\ge 5$ and it has no split.

\begin{figure}[t]
\centerline{\includegraphics[scale=0.42]{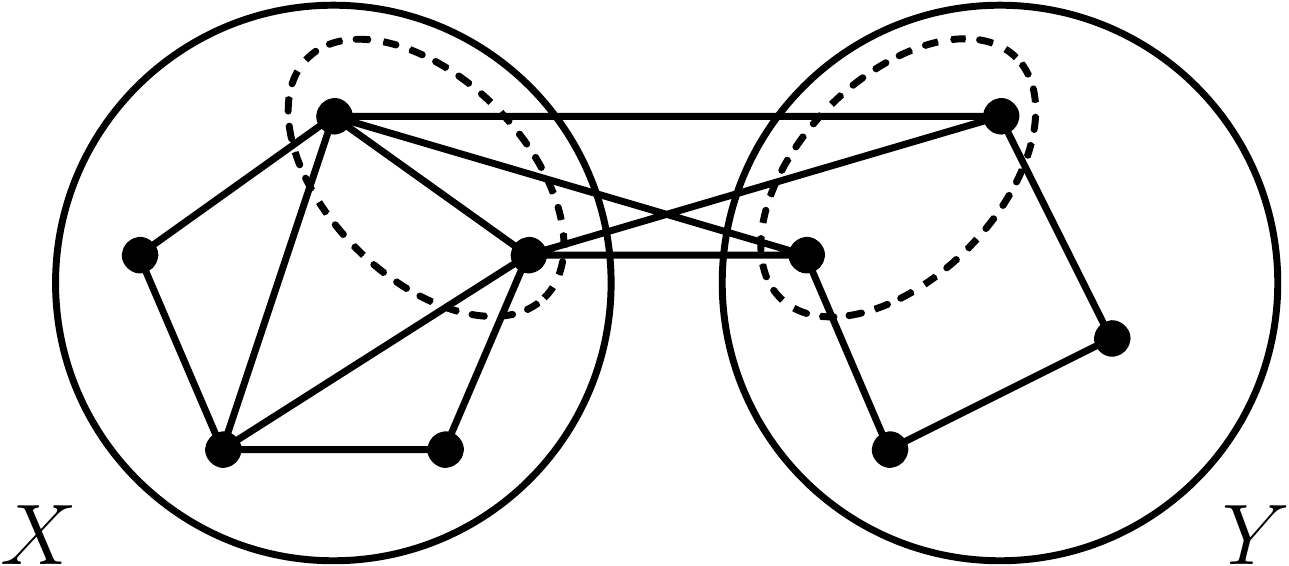}\qquad\qquad  \includegraphics[scale=0.42]{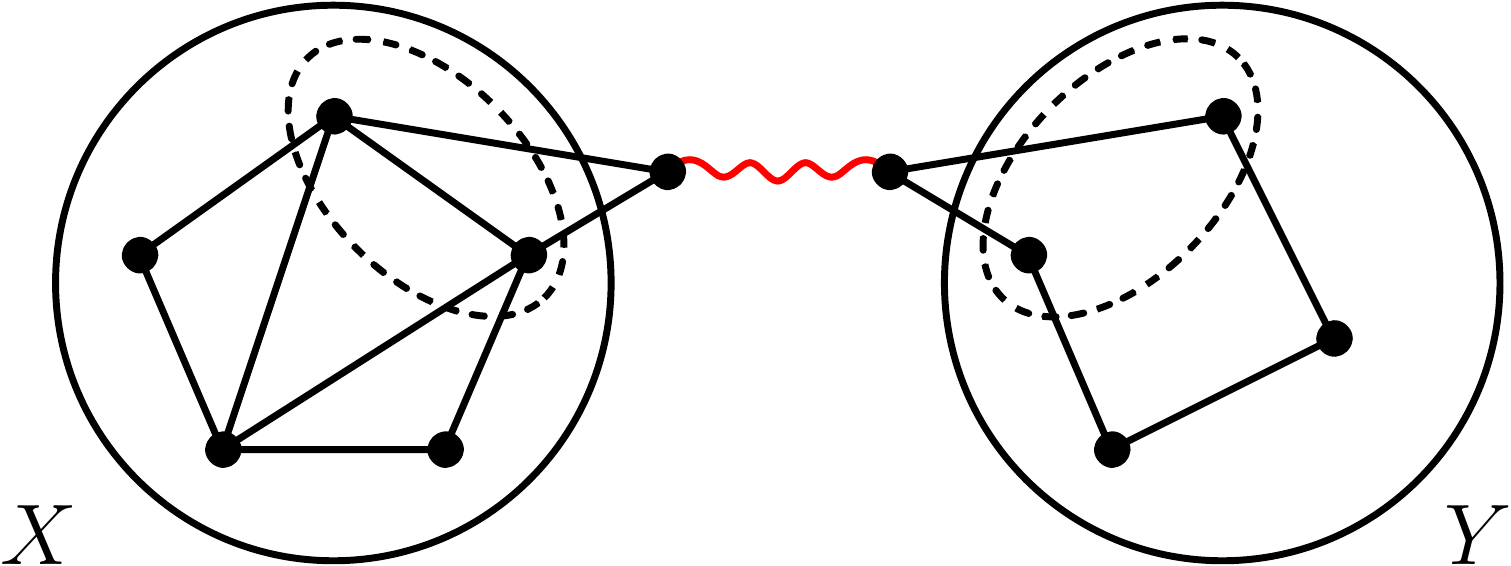}}
\caption{An example of a split $(X,Y)$ of a graph. Its simple decomposition is presented in the second picture, where the red edge is the newly introduced marked edge.}
\label{fig:examplesplit}
\end{figure}

A connected graph $D$ with a distinguished set of edges $M(D)$ is called a \emph{marked graph} if
the edges in $M(D)$ form a matching and each edge in $M(D)$ is a cut edge.
An edge in $M(D)$ is called a \emph{marked edge}, and every other edge is called an \emph{unmarked edge}.
A vertex incident with a marked edge is called a \emph{marked vertex},
and every other vertex is called an \emph{unmarked vertex}.
Each connected component of $D-M(D)$ is called a \emph{bag} of $D$.

When a connected marked graph $G$, which will be a bag of a marked graph, admits a split $(X,Y)$, we construct a marked graph $D$ on the vertex set $V(G) \cup \{x',y'\}$ such that
\begin{itemize}
\item for vertices $x,y$ with $\{x,y\}\subseteq X$ or $\{x,y\}\subseteq Y$, $xy\in E(G)$ if and only if $xy\in E(D)$,
\item $x'y'$ is a new marked edge,
\item $X$ is anti-complete to $Y$,
\item $\{x'\}$ is complete to $N_G(Y)$ and $\{y'\}$ is complete to $N_G(X)$ (with unmarked edges). 
\end{itemize}
The marked graph $D$ is called a \emph{simple decomposition of} $G$.
A \emph{split decomposition} of a connected graph $G$ is a marked graph $D$ defined inductively to be either $G$ or a marked graph defined from a split decomposition $D'$
of $G$ by replacing a connected component $H$ of $D'- M(D')$ with a simple decomposition of $H$.  See Figure~\ref{fig:example} for an example of a split decomposition.
The following lemma provides an important property. An example of an alternating path described in Lemma~\ref{lem:splitadj} is presented in Figure~\ref{fig:example}.

\begin{lemma}[See Adler, Kant\'e, and Kwon, Lemma 2.10 of \cite{AKK2014}]\label{lem:splitadj}
Let $D$ be a split decomposition of a connected graph $G$ and 
$u,v$ be two vertices in $G$.
Then $uv\in E(G)$ if and only if there is a path from $u$ to $v$ in $D$ where its first and last edges are unmarked, and
an unmarked edge and a marked edge alternatively appear in the path. 
\end{lemma}
 
Naturally, we can define a reverse operation of decomposing into a simple decomposition; for a marked edge $xy$ of a split decomposition $D$, 
\emph{recomposing $xy$} is the operation of removing two vertices $x$ and $y$ and making $N_D(x)\setminus \{y\}$ complete to $N_D(y)\setminus \{x\}$ with unmarked edges.
It is not hard to observe that if $D$ is a split decomposition of $G$, then $G$ can be obtained from $D$ by recomposing all marked edges.

Note that there are many ways of decomposing a complete graph or a star, because every its non-trivial vertex partition is a split.
Cunningham and Edmonds \cite{CunninghamE80} developed a canonical way to decompose a graph into a split decomposition by not allowing to decompose a bag which is a star or a complete graph.
A split decomposition $D$ of $G$ is called a \emph{canonical split decomposition} if each bag of $D$ is either a prime graph, a star, or a complete graph, and 
every recomposing of a marked edge in $D$ results in a split decomposition without the same property. It is not hard to observe that every canonical split decomposition has no marked edge linking two complete bags, and no marked edge linking a leaf of a star bag and the center of another star bag~\cite{Bouchet1988a}. Furthermore, for each pair of twins $a$ and $b$ in $G$, it holds that $a$ and $b$ must both be located in the same bag of the canonical split decomposition. 

\begin{figure}[t]
\centerline{\includegraphics[scale=0.66]{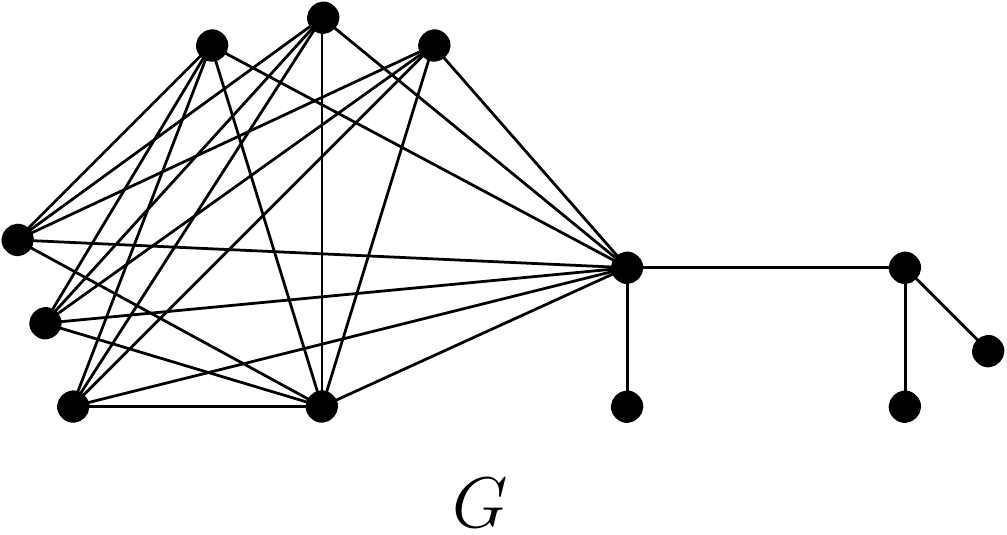}\qquad \includegraphics[scale=0.55]{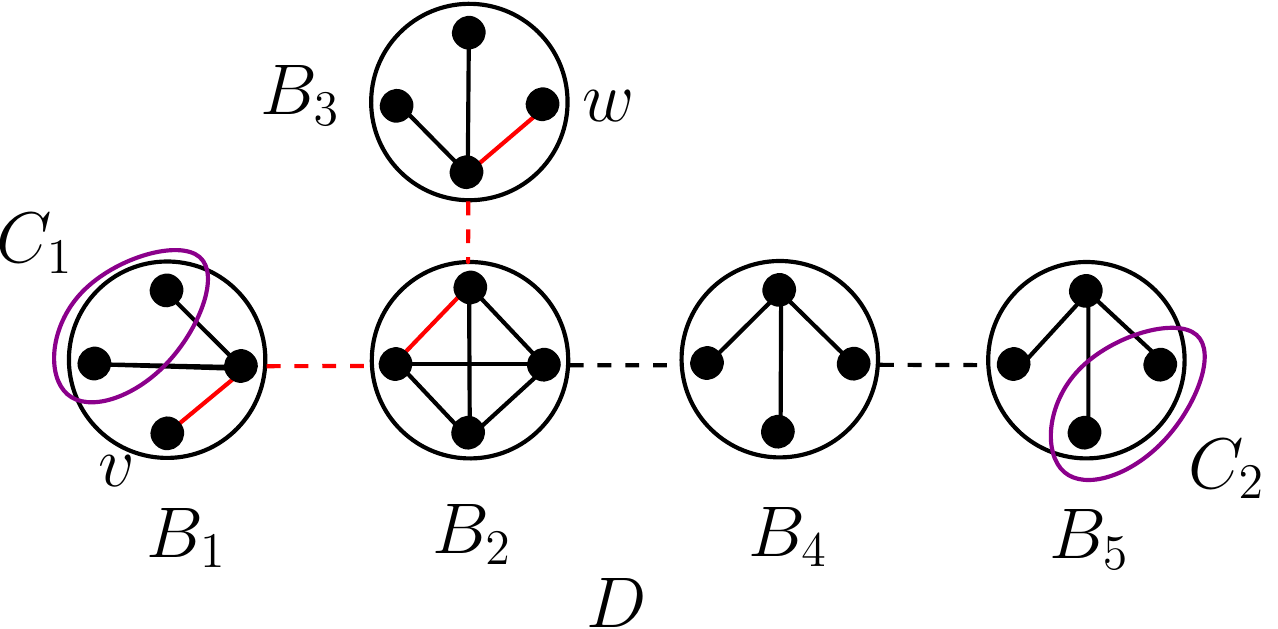}}
\caption{A graph $G$ and its canonical split decomposition $D$. Marked edges are represented by dashed edges, and bags are indicated by circles. 
Note that $\btwn(B_1, B_5)=\{B_1, B_2, B_4, B_5\}$ and $\{B_4, B_5\}$ is the set of $(C_1, C_2)$-separator bags, and $\{B_4\}$ is the set of $(B_1, B_5)$-separator bags. 
The shortest path from $v$ to $w$ in $D$ is a path from $v$ to $w$ where its first and last edges are unmarked, and
an unmarked edge and a marked edge alternatively appear in the path. The existence of such a path exactly corresponds to the adjacency relation in the original graph.
The distance between $C_1$ and $C_2$ in $G$ is $3$, and there are two $(C_1, C_2)$-separator bags.}
\label{fig:example}
\end{figure}

\begin{theorem}[Cunningham and Edmonds~\cite{CunninghamE80}] \label{thm:CED} 
Every connected graph has a unique canonical split decomposition, up to isomorphism.
\end{theorem}
\begin{theorem}[Dahlhaus~\cite{Dahlhaus00}]\label{thm:dahlhaus}
The canonical split decomposition of a graph  $G$ can be computed in time $\mathcal{O}(\abs{V(G)}+\abs{E(G)})$.
\end{theorem}

We can now give the second characterization of distance-hereditary graphs that is crucial for our results.
For convenience, we call a bag a \emph{star bag} or a \emph{complete bag} if it is a star or a complete graph, respectively.

\begin{theorem}[Bouchet~\cite{Bouchet1988a}]\label{thm:bouchet}
A graph is a distance-hereditary graph if and only if every bag in its canonical split decomposition is either a star bag or a complete bag.
\end{theorem}

We will later on also need a little bit of additional notation related to split decompositions of distance-hereditary graphs.
Let $D$ be a canonical split decomposition of a distance-hereditary graph.
For two distinct bags $B_1$ and $B_2$, we denote by $\comp (B_1, B_2)$ the connected component of $D-V(B_1)$ containing $B_2$.
Technically, when $B_1=B_2$, we define $\comp(B_1, B_2)$ to be the empty set.
For two bags $B_1$ and $B_2$, we denote by $\btwn (B_1, B_2)$ the set of all bags containing a vertex in a shortest path from $B_1$ to $B_2$ in $D$.
In other words, when we obtain a tree from $D$ by contracting every bag $B$ into a node $v(B)$,
$\btwn (B_1, B_2)$ is the set of all bags corresponding to nodes of the unique path from $v(B_1)$ to $v(B_2)$ in the tree. 
See Figure~\ref{fig:example}.

Let $C_1$ and $C_2$ be two disjoint vertex subsets of $D$ such that $C_1$ and $C_2$ are sets of unmarked vertices contained in (not necessarily distinct) bags $B_1$ and $B_2$, respectively. 
A bag $B$ is called \emph{a $(C_1, C_2)$-separator bag}
  if $B$ is a star bag contained in $\btwn (B_1, B_2)$ whose center is adjacent to neither $\comp (B, B_1)$ nor $\comp (B, B_2)$.
 We remark that $B$ can be $B_i$ for some $i\in \{1,2\}$, and especially when $B=B_1=B_2$, 
 $B$ is a star bag and each $C_i$ consists of leaves of $B$ and $B_1$ is the unique $(C_1, C_2)$-separator bag.
For convenience, we also say that a bag $B$ is \emph{a $(B_1, B_2)$-separator bag}
if  $B$ is a star bag contained in $\btwn (B_1, B_2)\setminus \{B_1, B_2\}$ whose center is adjacent to neither $\comp (B, B_1)$ nor $\comp (B, B_2)$.
For this notation, $B$ cannot be $B_1$ nor $B_2$.

We observe that the distance between $C_1$ and $C_2$ in the original graph 
  is exactly the same as one plus the number of $(C_1, C_2)$-separator bags. 

\begin{observation}\label{obser:separator}
The distance between $C_1$ and $C_2$ in the original graph 
  is exactly the same as one plus the number of $(C_1, C_2)$-separator bags. 
\end{observation}

\section{Setting the Stage}
\label{sec:general}

We begin by applying the \emph{iterative compression technique}, first introduced by Reed, Smith and Vetta~\cite{ReedSV2004} to show that \textsc{Odd Cycle Transversal} can be solved in single-exponential FPT time. This technique allows us to transform our original target problem to one that is easier to handle, which we call \disjointDHVD. Our goal for the majority of the paper will be to obtain a single-exponential FPT algorithm for \disjointDHVD; this is then used to obtain an algorithm for {\DHVD} in Section~\ref{sec:completing}.

    \smallskip
\noindent
\fbox{\parbox{0.97\textwidth}{
\disjointDHVD \\
\emph{Instance :} A graph $G$, an integer $k$, and $S\subseteq V(G)$ such that $G-S$ is distance-hereditary. \\
\emph{Task :} Is there $Q\subseteq V(G)\setminus S$ with $|Q|\le k$ such that $G-Q$ is distance-hereditary? }}
\vskip 0.2cm

We will denote an instance of {\disjointDHVD} as a tuple $(G,S,k)$.
The major part of our result is to prove that this problem can be solved in time $2^{\mathcal{O}(k+\cc(G[S]))}n^{\mathcal{O}(1)}$, 
where $\cc(G[S])$ denotes the number of connected components of $G[S]$.
We note that any instance of \disjointDHVD{} such that $G[S]$ is not distance-hereditary must clearly be a \NO-instance; hence we will assume that we reject all such instances immediately. 

Before explaining the general approach for solving \disjointDHVD, it will be useful to introduce a few definitions.
Since the canonical split decomposition guaranteed by Theorem~\ref{thm:bouchet} only helps us classify twins in $G-S$ and not in $G$, we explicitly define an equivalence $\sim$ on the vertices of $G-S$ which allows us to classify twins in $G$:
\begin{center}
for two vertices $u,v\in V(G-S)$, $u\sim v$ iff they are twins in $G$.
\end{center}

We denote by $\twinclass$ the set of equivalence classes of $\sim$ on $V(G-S)$, and each individual equivalence class will be called a \emph{twin class} in $G-S$.
We can observe that if $U\in \twinclass$ lies in a single connected component of $G-S$, then $U$ must be contained in precisely one bag of the split decomposition of this connected component of $G-S$, as $U$ is a set of twins in $G-S$ as well.
A twin class is \emph{$S$-attached} if it has a neighbor in $S$, and 
\emph{non-$S$-attached} if it has no neighbors in $S$.
Similarly, we say that a bag in the canonical split decomposition of $G-S$ is \emph{$S$-attached} if it has a neighbor in $S$, and
\emph{non-$S$-attached} otherwise.

We frequently use a special type of star bags.
A star bag $B$ is called \emph{simple} if 
its center is either unmarked or adjacent to a connected component of $D-V(B)$ consisting of one non-$S$-attached bag.

\subsection{Overview of the Approach}
Now that we have introduced the required terminology, we can provide a high-level overview of our approach for solving \disjointDHVD.
\begin{enumerate}
\item We exhaustively apply the branching rules described in Section~\ref{subsec:branching}. 
Branching rules will be applied when $G$ has a small subset $X\subseteq V(G-S)$ such that $S\cup X$ induces a DH obstruction, or 
there is a small connected subset $X\subseteq V(G-S)$ such that adding $X$ to $S$ decreases the number of connected components in $G[S]$.
\item We exhaustively apply the initial reduction rules described in Section~\ref{sec:rules}. 
Each of these rules runs in polynomial time, finds a part in the canonical split decomposition of a connected component of $G-S$ that can be simplified, and modifies the decomposition.
Each application of a reduction rule from Section~\ref{sec:rules} 
either reduces the number of vertices in $G-S$ or reduces the total number of bags in the canonical split decomposition (of a connected component of $G-S$). 
It is well known that the total number of bags in the canonical split decomposition of a graph is linear in the number of vertices.
Therefore, the total number of application of these initial reduction rules will also be at most linear in the number of vertices.
\item We say that $G$ and the canonical split decompositions of $G-S$ are \emph{reduced} if the branching rules in Section~\ref{subsec:branching} and reduction rules in Section~\ref{sec:rules} cannot be applied anymore.
We will obtain the following simple structure of the decompositions in the reduced instance:
\begin{itemize}
\item Each canonical split decomposition $D$ of a connected component of $G-S$ contains at least two distinct $S$-attached twin classes 
(Lemma~\ref{lem:onesattached}).
\item Each bag contains at most one $S$-attached twin class 
(Lemma~\ref{lem:twinclassreduction}).
\item When $B$ is a bag and $D'$ is a connected component of $D-V(B)$ containing no bags having a neighbor in $S$, 
$D'$ consists of one bag and $B$ is a star bag whose center is adjacent to $D'$ (Lemma~\ref{lem:smallbranch}).
In this case, $B$ is a simple star bag whose center is adjacent to $D'$.
\item When $B$ is a bag and $D'$ is a connected component of $D-V(B)$ such that
$D'$ contains exactly one $S$-attached bag $B'$,
there is no $(B', B)$-separator bag (Lemma~\ref{lem:simplifynearsattached2}).
\end{itemize}
\item Using these structures, we prove in Subsection~\ref{sec:twinclass} that if a split decomposition of a connected component of $G-S$ contains two $S$-attached twin classes, 
then one of reduction rules should be applied. 
For this, we assign any bag as a root bag $R$ of $D$ and 
choose a bag $B$ with maximum $\abs{\btwn(B,R)}$ such that there are two descendant bags of $B$ having $S$-attached twin classes $C_1$ and $C_2$, respectively.
Then the distance from $C_1$ to $C_2$ in $G-S$ is at most $2$, and thus their neighbors on $S$ should be close to each other, as branching rules cannot be applied further.
Depending on the type of $B$ and the distance from $C_1$ to $C_2$, 
we show separately that one of reduction rules can be applied.

It will imply that we can apply one of all rules recursively until $G-S$ is empty or $k$ becomes $0$.
Then we can test whether the resulting instance is distance-hereditary or not in polynomial time, and output an answer.
\end{enumerate}

Let us also say a few words about the running time of the algorithm. Let $\mu:=k+\cc(G[S])$. Each of our branching rules will reduce $\mu$ and branch into at most $6$ subinstances.
Each reduction rule takes polynomial time, and the reduction rules will be applied at most $\mathcal{O}(\abs{V(G)})$ times. 
Whenever we introduce a new rule, we need to show that it is \emph{safe}; for branching rules this means that there exists at least one subinstance resulting from the rule which is a \YES-instance if and only if the original graph was a \YES-instance, while for reduction rules this means that the application of the rule preserves the property of being a \YES-instance.

A vertex $v$ in $G-S$ is called \emph{irrelevant} if $(G, S, k)$ is a \YES-instance if and only if $(G-v, S, k)$ is a \YES-instance. We will be identifying and removing irrelevant vertices in several of our reduction rules. When removing a vertex $v$ from $G-S$, it is easy to modify the canonical split decomposition containing $v$, and thus it is not necessary to recompute the canonical split decomposition of the resulting graph from scratch.
More details regarding such modifications of split decompositions can be found in the work of Gioan and Paul~\cite{GP2012}.

\subsection{Branching Rules}\label{subsec:branching}
We state our two branching rules below.

\begin{BRULE}\label{brule:threevertices}
For every vertex subset $X$ of $G-S$ with $\abs{X}\le 5$, 
if $G[S\cup X]$  is not distance-hereditary, then we remove one of the vertices in $X$, and reduce $k$ by $1$.
\end{BRULE}
\begin{BRULE}\label{brule:reducecomponent}
For every vertex subset $X$ of $G-S$ with $\abs{X}\le 5$ such that $G[X]$ is connected and the set $N_G(X)\cap S$ is not contained in a connected component of $G[S]$, 
then we either remove one of the vertices in $X$ and reduce $k$ by $1$, or put all of them into $S$ (which reduces the number of connected components of $G[S]$).
\end{BRULE}

The safeness of Branching Rules~\ref{brule:threevertices} and \ref{brule:reducecomponent} are clear, and these rules can be performed in polynomial time.
The exhaustive application of these branching rules guarantees the following structure of the instance.

\begin{lemma}
\label{lem:shortdistance}
Let $(G, S, k)$ be an instance reduced under Branching Rules~\ref{brule:threevertices} and \ref{brule:reducecomponent}.
\begin{enumerate}[(1)]
\item $G$ has no small DH obstructions.
\item Let $v\in V(G-S)$. For every two vertices $x,y\in N_G(v)\cap S$, 
they are contained in the same connected component of $G[S]$ and there is no induced path of length at least $3$ from $x$ to $y$ in $G[S]$.
Specifically, if $xy\notin E(G)$, then there is an induced path $xpy$ for some $p\in S$.
\item There is no induced path $v_1 \cdots v_5$ of length $4$ in $G-S$ where $v_1$ and $v_5$ have neighbors in $S$ but $v_2$ and $v_4$ have no neighbors in $S$.
\item There is no induced path $v_1 \cdots v_4$ of length $3$ in $G-S$ where $v_1$ and $v_4$ have neighbors in $S$ but $v_2$ has no neighbors in $S$.
\end{enumerate}
\end{lemma}
\begin{proof}
(1) Suppose $G$ has a small DH obstruction $H$. Since $G-S$ is distance-hereditary, $V(H)\cap S\neq \emptyset$.
Thus, $\abs{V(H)\setminus S}\le 5$, and it can be reduced under Branching Rule~\ref{brule:threevertices}.

(2) First, by Branching Rule~\ref{brule:reducecomponent}, 
$x$ and $y$ are contained in the same connected component of $G[S]$.
Suppose there is an induced path of length at least $3$ from $x$ to $y$ in $G[S]$.
Then by Lemma~\ref{lem:dhobs}, $G[S\cup \{v\}]$ contains a DH obstruction, contradicting our assumption that $G$ is reduced under Branching Rule~\ref{brule:threevertices}.
So, if $xy\notin E(G)$, then there is an induced path of length $2$ from $x$ to $y$ in $G[S]$.

(3) Suppose there is an induced path $v_1 \cdots v_5$ of length $4$ in $G-S$ where $v_1$ and $v_5$ have neighbors on $S$ but $v_2$ and $v_4$ have no neighbors on $S$.
By Branching Rule~\ref{brule:reducecomponent}, 
we know that $N_G(v_1)\cap S$ and $N_G(v_5)\cap S$ are contained in the same connected component of $G[S]$.
Let $P$ be a shortest path from $N_G(v_1)\cap S$ to $N_G(v_5)\cap S$ (if $v_1$ and $v_5$ have a common neighbor, then we choose a common neighbor).
Then $v_2v_1Pv_5v_4$ is an induced path of length at least $4$ and $v_3$ is adjacent to its end vertices.
So, $G[S\cup \{v_1, \ldots, v_5\}]$ contains a DH obstruction, contradicting our assumption that $G$ is reduced under Branching Rule~\ref{brule:threevertices}.

(4) The same argument in (3) holds.
\end{proof}

Lemma~\ref{lem:shortdistance}, and especially point $(2)$ in the lemma, is used in many parts of our proofs. Since we will apply the branching rules exhaustively at the beginning and also after each new application of a reduction rule, these properties will be implicitly assumed to hold in subsequent sections.

We will make use of two more lemmas based on our branching rules. These
will be used in Section~\ref{sec:twinclass} as well as
in the proof of Lemma~\ref{lem:twinclassreduction} in Section~\ref{subsec:structure}.

\begin{lemma}
 \label{lem:relationt1t2}
Let $(G, S, k)$ be an instance reduced under Branching Rules~\ref{brule:threevertices} and \ref{brule:reducecomponent}.
Let $C_1, C_2$ be two distinct $S$-attached twin classes of $G-S$ such that $C_1$ is anti-complete to $C_2$, and $(N_G(C_1)\cap N_G(C_2))\cap V(G-S)\neq \emptyset$. Then:
\begin{enumerate}[(1)]
\item $(N_G(C_1)\cap N_G(C_2))\cap S\neq \emptyset$.
\item For every $x\in N_G(C_1)\setminus N_G(C_2)$ and every $y_1, y_2\in N_G(C_1)\cap N_G(C_2)$, 
if $x$ is adjacent to $y_1$, then $x$ is adjacent to $y_2$ as well.
It implies that $x$ is adjacent to either all of vertices in $N_G(C_1)\cap N_G(C_2)$ or neither of them.
\item For every $x\in N_G(C_1)\setminus N_G(C_2)$ and every $y_1, y_2\in N_G(C_1)\cap N_G(C_2)$, 
if there is a path $xpy_1$ for some $p\in S\setminus N_G(C_1)$ (not necessarily induced), then $p$ is adjacent to $y_2$ as well.
It implies that $p$ is adjacent to either all of vertices in $N_G(C_1)\cap N_G(C_2)$ or neither of them.
\end{enumerate} 
 \end{lemma}
 \begin{proof}
 For each $i\in \{1,2\}$ let $a_i\in C_i$ and let $T_i=N_G(C_i)$.

(1) Suppose $T_1\cap S$ and $T_2\cap S$ are disjoint.
Let us choose a vertex $z$ in $(T_1\cap T_2)\cap V(G-S)$, which is not an empty set by assumption.
Thus, $\{a_1, a_2, z\}$ induces a connected subgraph of $G$.
If $T_1\cap S$ and $T_2\cap S$ are not contained in one connected component of $G[S]$,
then we can apply Branching Rule~\ref{brule:reducecomponent}. 
As our instance was reduced under Branching Rule~\ref{brule:reducecomponent},
we know that $T_1\cap S$ and $T_2\cap S$ are contained in the same connected component of $G[S]$. 

Let $P$ be a shortest path from $T_1\cap S$ to $T_2\cap S$ in $G[S]$.
Clearly, $P$ contains at most one vertex from each $T_i\cap S$.
As $C_1$ is anti-complete to $C_2$, 
$a_1Pa_2$ is an induced path of length at least $3$ and $z$ is adjacent to its end vertices. 
By Lemma~\ref{lem:dhobs}, $G[V(P)\cup \{a_1, a_2, z\}]$ contains a DH obstruction, contradicting the assumption that
$G$ is reduced under Branching Rule~\ref{brule:threevertices}.

(2) For contradiction, suppose $xy_1\in E(G)$ and $xy_2\notin E(G)$.  Then $xa_1y_2a_2$ is an induced path of length $3$ and $y_1$ is adjacent to its end vertices.
By Lemma~\ref{lem:dhobs}, $G$ contains a small DH obstruction, contradiction. 

(3) Suppose there is a path $xpy_1$ for some $p\in S\setminus T_1$ and $p$ is not adjacent to $y_2$.   
First assume that $p\in S\setminus (T_1\cup T_2)$.
If $xy_2\in E(G)$, then $pxy_2a_2$ is an induced path, and
otherwise, $pxa_1y_2a_2$ is an induced path.
Since $y_1$ is adjacent to $p$ and $a_2$, by Lemma~\ref{lem:dhobs}, $G$ contains a small DH obstruction, contradiction. 
When $p\in (T_2\setminus T_1)\cap S$, $a_1xpa_2$ becomes an induced path of length $3$ and $y_1$ is adjacent to its end vertices, 
and thus $G$ contains a small DH obstruction. We conclude that $p$ is adjacent to $y_2$.
 \end{proof}

 \begin{lemma}\label{lem:completerelationt1t2}
 Let $(G, S, k)$ be an instance reduced under Branching Rules~\ref{brule:threevertices}, and \ref{brule:reducecomponent}.
Let $C_1, C_2$ be two distinct twin classes of $G-S$ such that $C_1$ is  complete to $C_2$. Then:
\begin{enumerate}[(1)]
\item For every $x\in N_G(C_1)\setminus (C_2\cup N_G(C_2))$ and every $y_1, y_2\in N_G(C_1)\cap N_G(C_2)$, 
if $x$ is adjacent to $y_1$, then either it is adjacent to $y_2$ as well, or $y_1$ is adjacent to $y_2$.
\item For every $x\in N_G(C_1)\setminus (C_2\cup N_G(C_2))$ and every $y\in N_G(C_1)\cap N_G(C_2)$, 
if there is a path $xpy$ for some $p\in V(G)\setminus (C_1\cup N_G(C_1))$ (not necessarily induced), then $p\in N_G(C_2)\setminus N_G(C_1)$.
\end{enumerate} 
 \end{lemma}
\begin{proof}
For each $i\in \{1,2\}$ and let $a_i\in C_i$ and let $T_i=N_G(C_i)$.

(1) Suppose $xy_1\in E(G)$ and $xy_2, y_1y_2\notin E(G)$.  Then $G[\{x,y_1, y_2, a_1, a_2\}]$ is isomorphic to the gem, contradiction.

(2) Suppose there is a path $xpy$ for some $p\in V(G)\setminus (C_1\cup N_G(C_1)\cup N_G(C_2))$.   
Then $pxa_1a_2$ is an induced path of length $3$, and $y$ is adjacent to its end vertices.
By Lemma~\ref{lem:dhobs}, 
$G[\{x,p,y, a_1, a_2\}]$ contains a small DH obstruction, contradiction.
\end{proof}

\section{Reduction Rules in Split Decompositions}\label{sec:rules}

In this section, we assume that the given instance $(G, S, k)$ is reduced under Branching Rules~\ref{brule:threevertices} and \ref{brule:reducecomponent}.
The reduction rules introduced here either remove some irrelevant vertex, move some vertex into $S$, or reduce the number of bags in the decomposition by modifying the instance into an equivalent instance.
After we apply any of these reduction rules, we will run the two branching rules from Section~\ref{sec:general} again.

In Subsection~\ref{subsec:bypassing}, we introduce the notion of a \emph{bypassing vertex}, which %
is a crucial concept that will frequently appear in our proofs.  
In Subsection~\ref{subsec:sixrules}, we present five reduction rules and prove their correctness.
Then in Subsection~\ref{subsec:structure}, we discuss structural properties of the obtained instance after exhaustive application of all of the presented branching rules and reduction rules.
These properties will be used in Section~\ref{sec:twinclass} to argue that if the instance is non-trivial, then one can apply one of reduction rules.

\subsection{Bypassing Vertices}\label{subsec:bypassing}

We introduce a generic way of finding an irrelevant vertex which will be used in many reduction rules.
For a vertex $v$ in $G-S$ and an induced path $H=p_1p_2p_3p_4p_5$ in $G$ where $p_3=v$,
a vertex $x$ in $S$ is called a \emph{bypassing vertex} for $H$ and $v$ if $x$ is adjacent to $p_2$ and $p_4$. 
When $H$ is clear from the context, we simply say that $x$ is a bypassing vertex for $v$. 
If such a vertex $x$ exists, it is clear that $x$ is not contained in $H$.
The following property is essential.

\begin{lemma}
\label{lem:badvertex}
Let $(G, S, k)$ be an instance reduced under Branching Rules~\ref{brule:threevertices} and \ref{brule:reducecomponent}.
Let $v$ be a vertex in $G-S$ such that for every induced path $P=p_1p_2p_3p_4p_5$ where $v=p_3$, 
there is a bypassing vertex for $P$ and $v$.
Then $v$ is irrelevant.
\end{lemma}
\begin{proof}
We claim that $(G, S, k)$ is a \YES-instance if and only if $(G-v, S, k)$ is a \YES-instance.
The forward direction is clear.
Suppose that $G-v$ has a vertex set $T$ such that $S\cap T\neq \emptyset$, $\abs{T}\le k$, and $(G-v)-T$ is distance-hereditary. 
If $G-T$ is distance-hereditary, then we are done.
Suppose that $G-T$ has a DH obstruction $H$. Since Branching Rule~\ref{brule:threevertices} does not apply, $G$ has no small DH obstructions, and therefore
$H$ is an induced cycle of length at least $7$. 
Let $P=p_1p_2p_3p_4p_5$ be the subpath of $H$ such that $p_3=v$.
By the assumption, there is a bypassing vertex $v'$ for $P$ and $v$.
Note that $v'\notin V(H)$, as $v'p_2p_3p_4v'$ would be a cycle of length $4$.
Also, $H-v$ is an induced path of length at least $5$. 
Thus $G[(V(H)\setminus \{v\})\cup \{v'\}]$ contains another DH obstruction by Lemma~\ref{lem:dhobs}. Note that $v'\in S$ and hence also $v'\in G-T$ and, in particular, $(V(H)\setminus \{v\})\cup \{v'\}\subseteq V((G-v)-T)$. 
This contradicts the fact that $(G-v)-T$ is distance-hereditary.
\end{proof}

\begin{lemma}
\label{lem:twovertexinS}
Let $(G, S, k)$ be an instance reduced under Branching Rules~\ref{brule:threevertices} and \ref{brule:reducecomponent}.
Let $v$ be a vertex in $G-S$ and $P=p_1p_2p_3p_4p_5$ be an induced path where $p_3=v$.
If $p_2, p_4\in S$, then there is a bypassing vertex for $P$ and $v$.
\end{lemma}
\begin{proof}
Note that $p_2p_4\notin E(G)$.
Thus, by (2) of Lemma~\ref{lem:shortdistance}, 
there is an induced path $p_2pp_4$ for some $p\in S$, and $p$ is a bypassing vertex.
\end{proof}

\subsection{Five Reduction Rules}\label{subsec:sixrules}

We are now ready to start with our reduction rules. 
For the remainder of this section, let us fix a canonical split decomposition $D$ of a connected component of $G-S$.

We start with a simple reduction rule that can be applied when $D$ contains at most one $S$-attached twin class.
\begin{RRULE}\label{rrule:dhcomponent}
If $D$ has at most one $S$-attached twin class, then we remove all unmarked vertices of $D$ from $G$.
\end{RRULE}
\begin{lemma}
\label{lem:onesattached}
Reduction Rule~\ref{rrule:dhcomponent} is safe.
\end{lemma}

\begin{proof}
If $D$ has no $S$-attached twin class, then its underlying graph is a distance-hereditary connected component of $G$.
Thus, we can safely remove all its unmarked vertices. We may assume that $D$ has one $S$-attached twin class $C$.

Since $C$ is the only $S$-attached twin class and every induced cycle of length at least $5$ contains no twins,
no induced cycle of length at least $5$ contains an unmarked vertex in $V(D)\setminus C$.
Thus, we can safely remove all of unmarked vertices other than vertices in $C$.
Now we assume that $V(D)=C$. We claim that every vertex in $C$ is also irrelevant.

To apply Lemma~\ref{lem:badvertex}, 
suppose there is an induced path
 $P=p_1p_2p_3p_4p_5$ where $p_3\in C$.
 Since there are no twins in $P$, 
 $P$ contains at most one vertex of $C$.
Thus, $p_2$ and $p_4$ are contained in $S$. Since $p_2p_4\notin E(G)$ and $(G,S,k)$ is reduced under Branching Rules~\ref{brule:threevertices} and \ref{brule:reducecomponent}, 
by Lemma~\ref{lem:twovertexinS}, there is a bypassing vertex for $P$ and $v$.
Since $P$ was arbitrarily chosen, by Lemma~\ref{lem:badvertex}, $v$ is irrelevant.
\end{proof}

The next rule deals with a vertex of degree $1$ in $G-S$.
 
\begin{RRULE}\label{rrule:leaftoS}
Let $B$ be a star bag whose center is unmarked, and let $v$ be a leaf unmarked vertex in $B$. 
If $v$ has no neighbor in $S$, then we remove $v$.
If $v$ has a neighbor in $S$, then we move $v$ into $S$.
\end{RRULE}

\begin{lemma}
\label{lem:leaftoS}
Reduction Rule~\ref{rrule:leaftoS} is safe.
\end{lemma}
\begin{proof}
Let $x$ be the center of $B$ and let $v$ be a leaf unmarked vertex in $B$.
If $v$ has no neighbor in $S$, then $v$ has degree $1$ in $G$, and we can safely remove it.
We assume that $v$ has a neighbor in $S$.
We claim that $(G, S, k)$ is a \YES-instance if and only if $(G, S\cup \{v\}, k)$ is a \YES-instance.
The converse direction is easy. Suppose that $G$ contains a vertex set $T$ where $S\cap T=\emptyset$, $\abs{T}\le k$, and $G-T$ is distance-hereditary. 
Let $T'=T$ if $v\notin T$, and otherwise, we remove $v$ from $T$ and add $x$ to $T$, and call it $T'$.
We claim that 
$G-T'$ is distance-hereditary, which implies that $(G, S\cup \{v\}, k)$ is a \YES-instance.

Suppose $G-T'$ is not distance-hereditary.
Since $G$ has no small DH obstructions, $G-T'$ contains an induced cycle $H$ of length at least $7$.
First assume that $H$ contains $x$. 
Then $x$ is not contained in $T'$, and therefore, $v$ was not contained in $T$, and we have $T=T'$ by the construction. Thus $G-T$ also contains $H$, contradiction.
Thus, we have $x\notin V(H)$. 
If $v\notin V(H)$, then $H$ is an induced subgraph of $G-T$ because $T$ and $T'$ only differ at $\{v,x\}$.
This implies that $v\in V(H)$ and $v\in T$. Thus $T'$ contains $x$.

Let $P=p_1p_2p_3p_4p_5$ be the subpath of $H$ where $p_3=v$.
As $T'$ contains $x$, 
$p_2$ and $p_4$ are contained in $S$.
As $(G,S,k)$ is reduced under Branching Rules~\ref{brule:threevertices} and \ref{brule:reducecomponent}, 
by Lemma~\ref{lem:twovertexinS}, there is a bypassing vertex for $P$ and $v$.
Thus, $(G-T')[(V(H)\setminus \{v\})\cup S]$ contains another DH obstruction.
It contradicts the fact that $G-T$ is distance-hereditary, as $(G-T')[(V(H)\setminus \{v\})\cup S]$ is an induced subgraph of $G-T$.
\end{proof}

We remark that when we move $v$ into $S$ in Reduction Rule~\ref{rrule:leaftoS}, 
$k+\cc(G[S])$ does not increase, and the size of $V(G)\setminus S$ decreases.
After applying Reduction Rule~\ref{rrule:leaftoS} exhaustively, we obtain that 
if the center of a star bag is unmarked, then this bag contains no unmarked leaves.

The next reduction rule arises directly from the definition of bypassing vertices.

\begin{RRULE}\label{rrule:p5middle}
Let $v$ be a vertex in $G-S$ such that 
for every induced path $P=p_1p_2p_3p_4p_5$ with $p_3=v$, 
there is a bypassing vertex for $P$ and $v$.
Then we remove $v$. In  particular, when there is no such an induced path, we remove $v$.
\end{RRULE}

For fixed $v$, we can apply Reduction Rule~\ref{rrule:p5middle} in time $\mathcal{O}(\abs{V(G)}^5)$ by considering all vertex subsets of size $4$, 
and testing whether $p_2$ and $p_4$ have a common neighbor in $S$.

\begin{lemma}
  \label{lem:p5middle}
Reduction Rule~\ref{rrule:p5middle} is safe.
\end{lemma}
\begin{proof}
This follows from Lemma~\ref{lem:badvertex}.
\end{proof}

We proceed by introducing a reduction rule which sequentially arranges bags containing exactly one twin class.
The operation of \emph{swapping the adjacency} between two vertices $x$ and $y$ in a graph is to remove $xy$ if $xy$ was an edge, and otherwise add an edge between $x$ and $y$.
The number of bags in $D$ is strictly reduced when applying Reduction Rule~\ref{rrule:leaf}.

\begin{RRULE}\label{rrule:leaf}
Let $B$ be a leaf bag and $B'$ be the neighbor bag of $B$.
\begin{enumerate}[(1)]
\item If $B$ is a complete bag having exactly one twin class in $G-S$ and $B'$ is a star bag whose leaf is adjacent to $B$, 
then we swap the adjacency between every two unmarked vertices in $B$. By swapping the adjacency, $B$ becomes a star whose center is adjacent to $B'$, and thus we can recompose the marked edge connecting $B$ and $B'$.
We recompose the marked edge connecting $B$ and $B'$.
\item If $B$ is a star bag having exactly one twin class in $G-S$, the center of $B$ is adjacent to $B'$, and $B'$ is a complete bag,  
then we swap the adjacency between every two unmarked vertices in $B$. By swapping the adjacency, $B$ becomes a complete graph, and thus we can recompose the marked edge connecting $B$ and $B'$.
We recompose the marked edge connecting $B$ and $B'$.
\end{enumerate}
\end{RRULE}

We use the following lemma.
\begin{lemma}\label{lem:swaptwinproperty}
Let $A$ be a set of vertices that are pairwise twins in $G$. Let $G'$ be the graph obtained from $G$ by swapping the adjacency relation between every pair of two distinct vertices in $A$.
Then $(G, S, k)$ is a \YES-instance if and only if $(G', S, k)$ is a \YES-instance.
\end{lemma}
\begin{proof}
Note that either $G[A]$ is a complete graph or it has no edges. Therefore, $A$ is again a set of vertices that are pairwise twins in $G'$.
Since each DH obstruction contains at most one vertex from a set of twins (and hence, at most one vertex from $A$), swapping the adjacency on $A$ will neither introduce nor remove DH obstructions from $G$.
Hence it is easy to check that  $(G, S, k)$ is a \YES-instance if and only if $(G', S, k)$ is a \YES-instance.
\end{proof}
\begin{lemma}
Reduction Rule~\ref{rrule:leaf} is safe.
\end{lemma}
\begin{proof}
This follows from Lemma~\ref{lem:swaptwinproperty}.
\end{proof}

The last rule consider bags near to some leaf bag.
We illustrate in Figure~\ref{fig:bypassing2}.
Recall that 
a star bag $B$ is simple if 
its center is either unmarked or adjacent to a connected component of $D-V(B)$ consisting of one non-$S$-attached bag.

\begin{RRULE}\label{rrule:bypassing2}
Let $B_1, B_2, B_3$ be distinct bags in $D$ such that  
\begin{itemize}
\item $B_1$ is a non-$S$-attached leaf bag whose neighbor bag is $B_2$, and it is not a star whose leaf is adjacent to $B_2$, 
\item $B_2$ has exactly two neighbor bags $B_1$ and $B_3$, it is a star whose center is adjacent to $B_1$, and the set of unmarked vertices in $B_2$ is the unique $S$-attached twin class $C_2$ in $B_2$,  and
\item $B_3$ is a simple star bag.
\end{itemize}
Let $C_1$ be the set of unmarked vertices in $B_1$.
Then we remove $B_1$ and $B_2$, and
add a leaf set of unmarked vertices $\widetilde{C}$ with $\min ( \abs{C_1}, \abs{C_2})$ vertices to $B_3$, that is complete to $N_G(C_2)\cap S$ and has no other neighbors in $S$.
\end{RRULE}

Note that this rule can potentially create an induced cycle of length $6$. So, we need to run Branching Rule~\ref{brule:threevertices} after applying Reduction Rule~\ref{rrule:bypassing2}. 
We confirm the safeness of Reduction Rule~\ref{rrule:bypassing2} in the following lemma.

  \begin{figure}[t]
    
    \centering
      \includegraphics[scale=0.55]{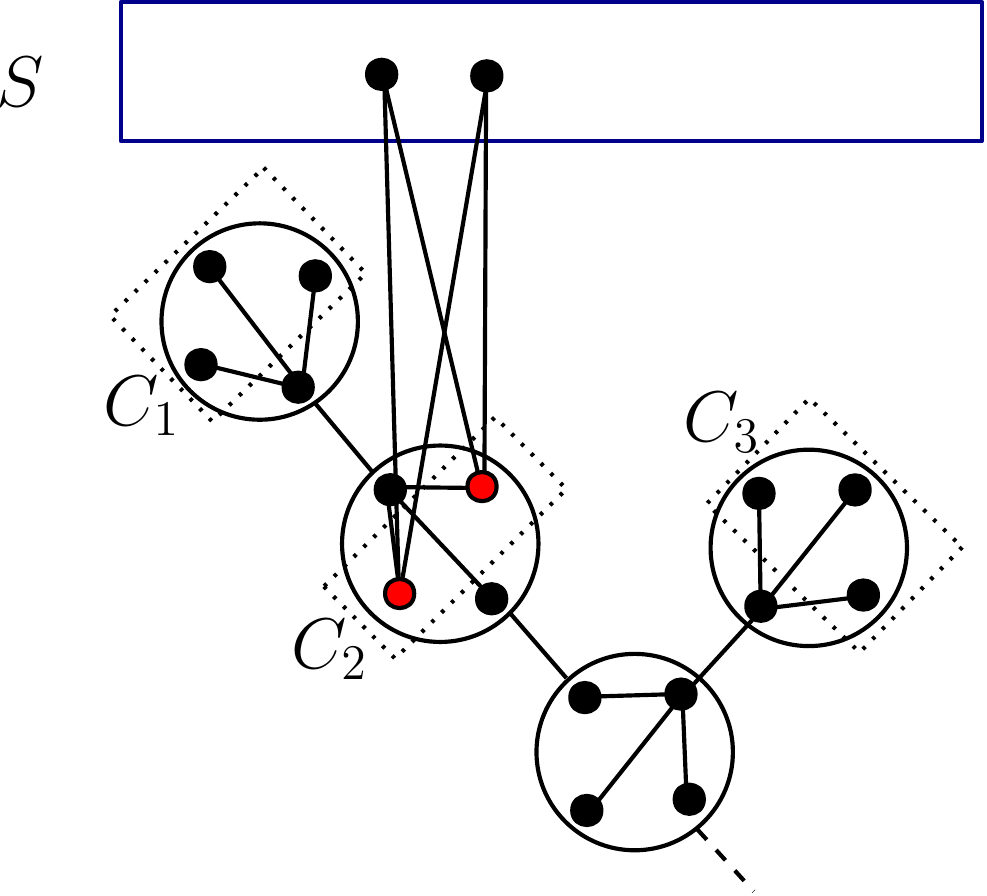}
    \qquad
      \includegraphics[scale=0.55]{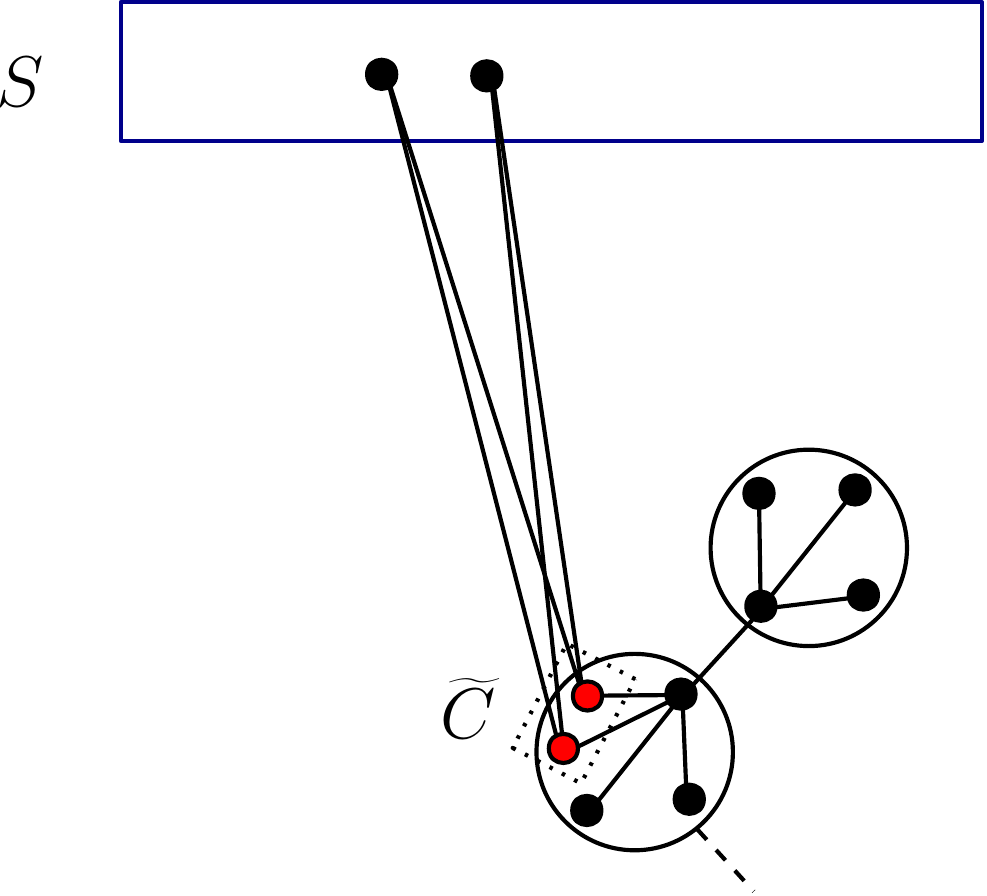}
        
    \caption{Reduction Rule~\ref{rrule:bypassing2}.} \label{fig:bypassing2}
      
  \end{figure}

\begin{lemma}
\label{lem:bypassing2}
Reduction Rule~\ref{rrule:bypassing2} is safe.
\end{lemma}
\begin{proof}
As $B_3$ is a simple star bag and $C_1$ has no neighbors in $S$, 
$N_G(C_1)\setminus C_2$ is exactly 
the center of $B_3$ if it is unmarked, and otherwise, the set of unmarked vertices in the bag where the center of $B_3$ is adjacent.
Let $C_3=N_{G}(C_1)\setminus C_2$. 
We remark that $C_3$ is a twin class. 

Let $G'$ be the resulting graph obtained by applying Reduction Rule~\ref{rrule:bypassing2}.
Note that $\widetilde{C}$ is a set of pairwise twins in $G'$ (it may not be a twin class), 
and $G-(C_1\cup C_2)=G'-\widetilde{C}$.

We claim that $(G, S, k)$ is a \YES-instance if and only if $(G', S, k)$ is a \YES-instance.
Suppose $G$ has a minimum vertex set $T$ such that $\abs{T}\le k$, $S\cap T=\emptyset$, and $G-T$ is distance-hereditary.
We divide cases depending on whether $T$ contains a vertex of $C_1\cup C_2$ or not.

\subparagraph{\textbf{Case 1.} $T$ contains a vertex in $C_1\cup C_2$:}  
We observe that since $C_i$ is a twin class and $T$ is a minimum solution, 
if $T$ contains a vertex of $C_i$, then $T$ contains all vertices in $C_i$.
Thus,  $T$ fully contains one of $C_1$ and $C_2$.
Since $\widetilde{C}=\min(\abs{C_1}, \abs{C_2})$, the set $T'=(T\setminus (C_1\cup C_2))\cup \widetilde{C}$ has size at most $k$. Moreover, we conclude that
$G'-T'$ is distance-hereditary, as it is an induced subgraph of $G-T$.

\subparagraph{\textbf{Case 2.} $T$ contains no vertex in $C_1\cup C_2$:}  

Suppose that $G'-T$ contains a DH obstruction $H$.
If $H$ does not contain a vertex in $\widetilde{C}$, then $H$ is an induced subgraph of $G-T$, contradicting our assumption.
Thus, $H$ contains a vertex in $\widetilde{C}$, and as every pair of two distinct vertices in $\widetilde{C}$ is a twin, we have $\abs{V(H)\cap \widetilde{C}}=1$.
Let $v$ be the vertex in $V(H)\cap \widetilde{C}$, and let $w,z$ be the two neighbors of $v$ in $H$. As $C_3$ is a twin class in $G'-S$, 
at least one of $w$ and $z$ is contained in $S$. Without loss of generality, we assume $w\in S$.

If $z\in S$, then we can obtain a DH obstruction by replacing $v$ with a vertex of $C_2$ in $G$,
which implies that $G-T$ contains a DH obstruction. Thus, we may assume that $z$ is contained in $V(G-S)$, and henceforth we have $z\in C_3$.
For two vertices $c_1\in C_1$ and $c_2\in C_2$, 
we can obtain a DH obstruction in $G-T$ from $H$ by removing $v$ and adding $c_1,c_2$, which is equivalent (up to isomorphism) to subdividing the unique edge in $H$ incident to $v$ and a vertex in $C_3$. By Observation~\ref{obs:sub}, we know that the resulting graph $G-T$ must then also contain a DH obstruction, contradicting our assumption.

\medskip

For the converse direction, suppose that $G'$ has a minimum vertex set $T'$ such that $\abs{T'}\le k$, $S\cap T'=\emptyset$, and $G'-T'$ is distance-hereditary.
Similar to the forward direction, we divide cases depending on whether $T'$ contains a vertex in $\widetilde{C}$ or not.

\subparagraph{\textbf{Case 1.} $T'$ contains no vertex in $\widetilde{C}$:}  
Suppose $G-T'$ has a DH obstruction $H$. Since $G$ has no small DH obstructions due to the application of branching rules, $H$ should be an induced cycle of length at least $7$.
We have $V(H)\cap (C_1\cup C_2)\neq \emptyset$, otherwise $H$ is an induced subgraph of $G'-T'$, which is contradiction.
As $C_1$ and $C_2$ are twin classes, $H$ contains at most one vertex from each of $C_1$ and $C_2$.

We claim that $H$ contains one vertex from each of $C_1$ and $C_2$.
Suppose $V(H) \cap C_1 \neq \emptyset$ and $V(H) \cap C_2 = \emptyset$. Then the two neighbors of the vertex on $C_1\cap V(H)$ belong to $C_3$, 
since $C_3 = N_G(C_1) \setminus C_2$. But $C_3$ forms a twin class, and an induced cycle of length at least $7$ cannot contain two vertices from the same twin class; a contradiction.
Suppose $V(H)\cap C_1= \emptyset$ but $V(H)\cap C_2\neq \emptyset$.  
Then the two neighbors of the vertex $v$ in $V(H)\cap C_2$ in $H$ are contained in $S$.
Let $P=p_1p_2p_3p_4p_5$ be the subpath of $H$ where $p_3=v$.
By Lemma~\ref{lem:twovertexinS}, there is a bypassing vertex for $P$ and $v$, and thus $G[(V(H)\setminus \{v\})\cup S]$ contains a DH obstruction, which is also contained in $G'-T'$.
This constitutes a contradiction. We conclude that $H$ contains one vertex from each of $C_1$ and $C_2$.

 It further implies that $H$ contains one vertex from each of $C_3=N_G(C_1)\setminus C_2$ and $N_G(C_2)\cap S$, because $N_G(C_2)\setminus C_1\subseteq S$. 
Since $H$ has length at least $7$, we can obtain an induced cycle of length at least $6$ in $G'-T$ from $H$ by removing the vertices in $C_1\cup C_2$ and adding one vertex of $\widetilde{C}$, which is contradiction.

\subparagraph{\textbf{Case 2.} $T'$ contains a vertex in $\widetilde{C}$:}  
As $\widetilde{C}$ is a twin class and $T'$ is a minimum solution for $G'$, we have $\widetilde{C}\subseteq T'$.
We obtain a set $T$ from $T'$ by removing $\widetilde{C}$, and adding $C_1$ if $\abs{C_1}=\abs{\widetilde{C}}$ and adding $C_2$ if $\abs{C_2}=\abs{\widetilde{C}}$. If $\abs{C_1}=\abs{C_2}$, then we add one of them chosen arbitrarily.
Clearly, $\abs{T}\le \abs{T'}\le k$.
We claim that $G-T$ is distance-hereditary.

In case when $C_2\subseteq T$, we observe that every induced cycle of length at least $7$ containing a vertex in $C_1$ has to contain two vertices in $C_3$, which is not possible.
Thus, we may assume $C_1\subseteq T$.
Note that $N_G(C_2)\subseteq S\cup C_1$. Thus, whenever there is an induced cycle of length at least $7$ in $G-T$ containing a vertex in $C_2$, by Lemma~\ref{lem:twovertexinS} there exists another DH obstruction which does not contain any vertex in $C_2$, contradicting the assumption that $G'-T'=G-(T\cup C_2)$ is distance-hereditary.
Hence we conclude that $G-T$ is distance-hereditary.
\end{proof}

\begin{PROP}
\label{prop:applicationRR}
 Let $(G, S, k)$ be an instance reduced under Branching Rules~\ref{brule:threevertices} and \ref{brule:reducecomponent}.
Given a connected component $H$ of $G-S$, we can in time $\mathcal{O}(\abs{V(G)}^6)$ either apply one of Reduction Rules~\ref{rrule:dhcomponent}--\ref{rrule:bypassing2}, or correctly answer that Reduction Rules~\ref{rrule:dhcomponent}--\ref{rrule:bypassing2} cannot be applied anymore. 
\end{PROP}
\begin{proof}
We first compute the canonical split decomposition $D$ of $H$ in time  $\mathcal{O}(\abs{V(G)}+\abs{E(G)})$ using Theorem~\ref{thm:dahlhaus}.
Then we classify twin classes in $D$ by testing two unmarked vertices in a bag have the same neigbhorhood in $S$ or not. 
This can be done in time $\mathcal{O}(\abs{V(G)}^2)$. 
At the same time, we can also test whether a twin class is $S$-attached or not. 
Note that the total number of bags in canonical split decompositions of connected components of $G-S$ is $\mathcal{O}(\abs{V(G)})$.

We can apply Reduction Rules~\ref{rrule:dhcomponent}, \ref{rrule:leaftoS}, \ref{rrule:leaf} in time $\mathcal{O}(\abs{V(G)})$, if one of them can be applied. 
We can apply Reduction Rule~\ref{rrule:p5middle} in time $\mathcal{O}(\abs{V(G)}^5)$ for fixed vertex $v$, 
and thus, we can test for all vertices $v\in V(G)\setminus S$ in time $\mathcal{O}(\abs{V(G)}^6)$.
For Reduction Rule~\ref{rrule:bypassing2}, we need to consider three bags, which are uniquely identified by the first (leaf) bag among them, to check whether they satisfy preconditions of the rule. We can verify the preconditions of Reduction Rule~\ref{rrule:bypassing2} in constant time and 
thus this step takes time $\mathcal{O}(\abs{V(G)})$.
We conclude that we can in time $\mathcal{O}(\abs{V(G)}^6)$ either apply one of Reduction Rules~\ref{rrule:dhcomponent}--\ref{rrule:bypassing2}, or correctly answer that Reduction Rules~\ref{rrule:dhcomponent}--\ref{rrule:bypassing2} cannot be applied anymore. 
\end{proof}

\subsection{Structural Properties obtained after Exhaustive Application of Rules}\label{subsec:structure}

In this subsection, we discuss structural properties obtained after the exhaustive application of both branching and reduction rules.
We say that $G$ and the canonical split decompositions of connected components of $G-S$ are \emph{reduced} 
if Branching Rules~\ref{brule:threevertices}--\ref{brule:reducecomponent} and Reduction Rules~\ref{rrule:dhcomponent}--\ref{rrule:bypassing2} cannot be applied anymore.
We assume that the given instance is reduced in this subsection.

The following observation is a direct consequence of the exhaustive application 
of Reduction Rule~\ref{rrule:leaftoS}.

\begin{observation}\label{obs:starrestriction}
If the center of a star bag in $D$ is unmarked, then this bag contains no unmarked leaves.
\end{observation}

Our next goal is to establish the following lemma.

\begin{lemma}\label{lem:twinclassreduction}
Every bag of $D$ contains at most one $S$-attached twin class.
\end{lemma}

Before we formally prove Lemma~\ref{lem:twinclassreduction}, we briefly explain how the argument works.
Let $C_1$ and $C_2$ be two distinct $S$-attached twin classes in a bag $B$ such that neither of them consists of the center of a star and $(N_G(C_1)\setminus N_G(C_2))\cap S$ is non-empty.
If $B$ is a star, then $C_1$ is anti-complete to $C_2$ and $C_1$, $C_2$ have a common neighbor in $G-S$, and thus, $C_1$ and $C_2$ satisfy preconditions of Lemma~\ref{lem:relationt1t2}.
Lemma~\ref{lem:relationt1t2} implies that $(N_G(C_1)\cap N_G(C_2))\cap S$ is non-empty. 
Let $x\in (N_G(C_1)\setminus N_G(C_2))\cap S$ and $y\in (N_G(C_1)\cap N_G(C_2))\cap S$. 
We argue that whenever there is an induced path $w'wvzz'$ with $v\in C_1$ there is a bypassing vertex for $v$.
We describe an example case. For instance, when $w$ and $z$ are both contained in $V(G-S)$, they are contained in $N_G(C_1)\cap N_G(C_2)$. 
So, $x\in N_G(C_1)\setminus N_G(C_2)$ while $w, z ,y\in N_G(C_1)\cap N_G(C_2)$.
Thus, by (2) of Lemma~\ref{lem:relationt1t2},
if $x$ is adjacent to $y$, then $x$ should be adjacent to $w$ and $z$, which means that $x$ is a bypassing vertex for $v$.
If $x$ is not adjacent to $y$, then we could apply (3) of Lemma~\ref{lem:relationt1t2} to find a bypassing vertex for $v$.
We do a careful analysis depending on the places of $w$ and $z$, and also consider the case when $B$ is a complete bag.

 \begin{proof}[Proof of Lemma~\ref{lem:twinclassreduction}]
Suppose there is a bag containing two distinct $S$-attached twin classes $C_1$ and $C_2$.
By Observation~\ref{obs:starrestriction}, if $C_i$ consists of the center of a star, then there are no other unmarked vertices in the bag, and thus it is not possible.
Therefore, $C_i$ does not consist of the center of a star bag.
As $C_1$ and $C_2$ are distinct twin classes, $N_G(C_1)\cap S\neq N_G(C_2)\cap S$, and thus we have 
either $(N_G(C_1)\setminus N_G(C_2))\cap S\neq \emptyset$ or $(N_G(C_2)\setminus N_G(C_1))\cap S\neq \emptyset$.
Without loss of generality, we assume $(N_G(C_1)\setminus N_G(C_2))\cap S$ is non-empty.

For each $i\in \{1,2\}$, let $c_i\in C_i$ and let $T_i=N_G(C_i)\setminus C_{3-i}$. Let $x\in (T_1\setminus T_2)\cap S$.
We observe that $(T_1\setminus T_2)\cap V(G-S)=(T_2\setminus T_1)\cap V(G-S)=\emptyset$.
This is because $C_1$ and $C_2$ are contained in $B$, which is a complete bag or a star bag whose center is marked.

We claim that for every $v\in C_1$ and every induced path $H=w'wvzz'$,
there is a bypassing vertex for $H$ and $v$. 
This will imply that we can apply Reduction Rule~\ref{rrule:p5middle}, which leads a contradiction.

If $w, z\in S$, then by Lemma~\ref{lem:twovertexinS}, there is a bypassing vertex for $v$.
We may assume that $w$ or $z$ is contained in $G-S$.
Without loss of generality, we assume that $w$ is contained in $G-S$.
We depict cases in Figures~\ref{fig:lemma45star1} and \ref{fig:lemma45star2}.

  \begin{figure}[t]
\begin{subfigure}[b]{0.5\textwidth}
      \includegraphics[scale=0.5]{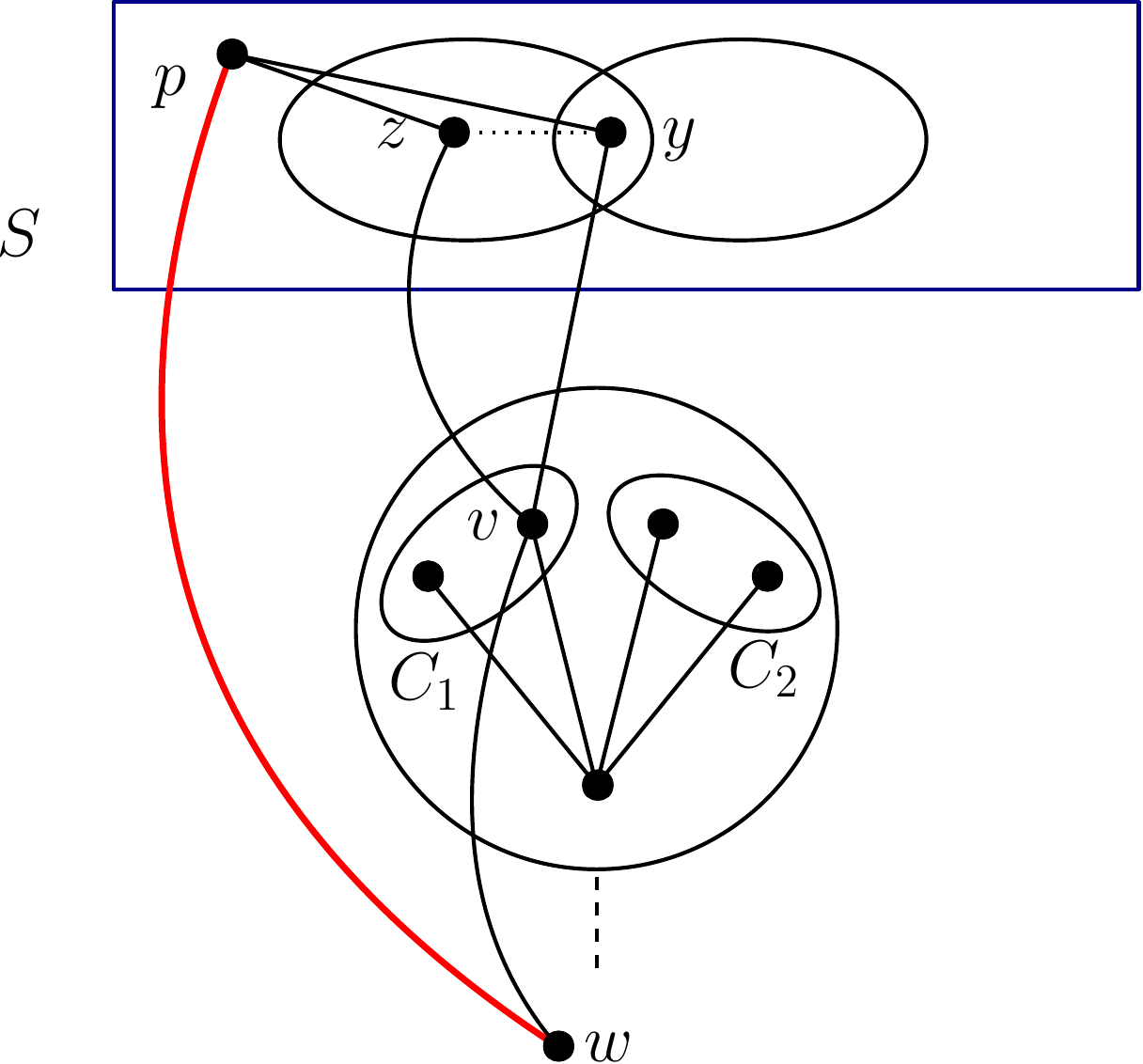}
      \caption{$z\in (T_1\setminus T_2)\cap S$}
      \end{subfigure}
      \qquad
\begin{subfigure}[b]{0.5\textwidth}
      \includegraphics[scale=0.5]{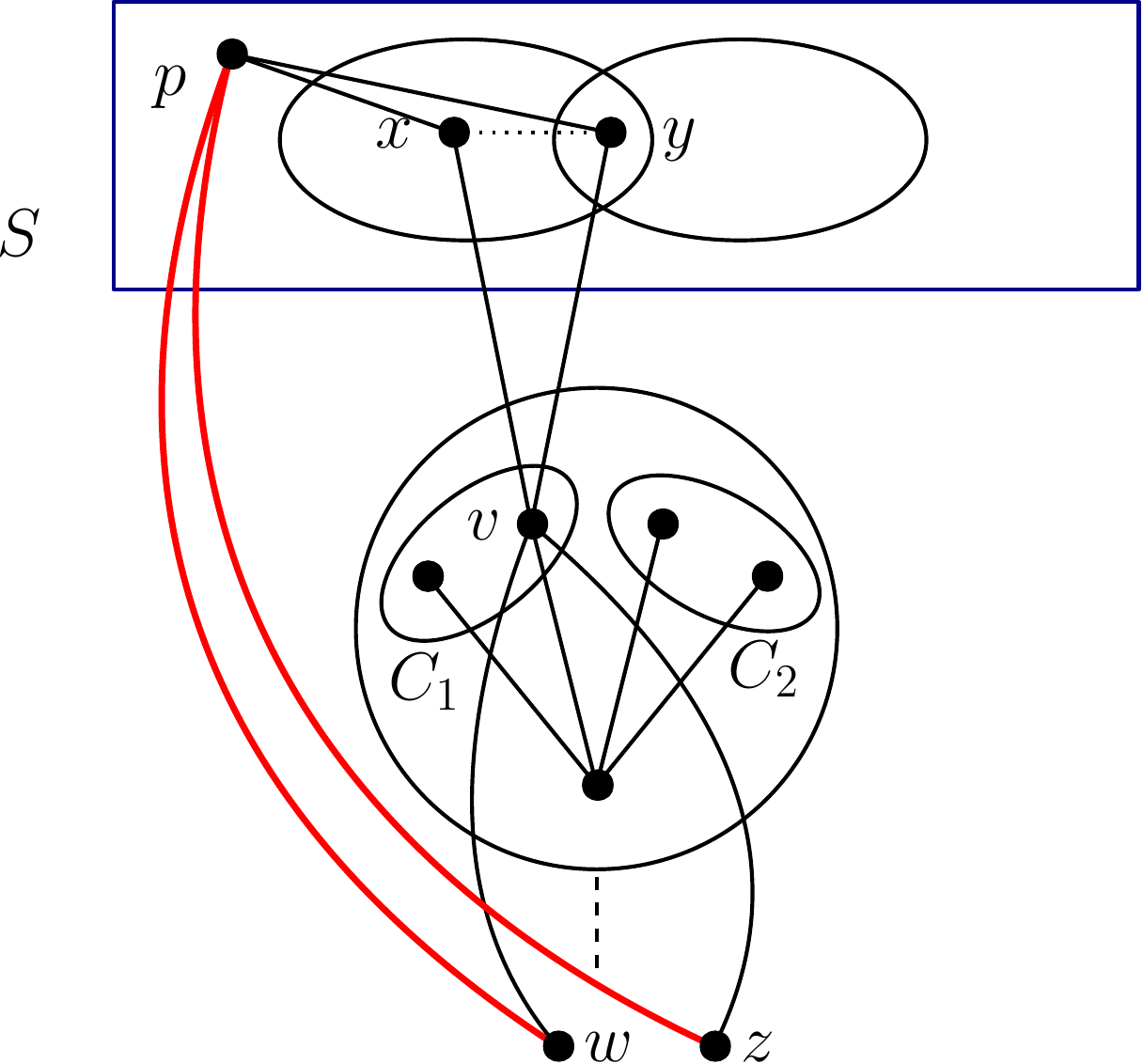}
      \caption{$z\in (T_1\cap T_2)\cap V(G-S)$}
      \end{subfigure}
            \caption{When $B$ is a star bag in Lemma~\ref{lem:twinclassreduction}.
            The red thick edges illustrate the edges whose existence is guaranteed by Lemma~\ref{lem:relationt1t2}. } \label{fig:lemma45star1}
  \end{figure}

\subparagraph{\textbf{Case 1.} $B$ is a star bag:}  
In this case, $C_1$ is anti-complete to $C_2$ and $w\in (T_1\cap T_2)\cap V(G-S)$.
By (1) of Lemma~\ref{lem:relationt1t2}, we have $(T_1\cap T_2)\cap S\neq \emptyset$. Let $y\in (T_1\cap T_2)\cap S$. 
We divide cases depending on whether $z\in (T_1\setminus T_2)\cap S$ or $z\in T_1\cap T_2$.

Suppose $z\in (T_1\setminus T_2)\cap S$. Note that both $y$ and $w$ are contained in $T_1\cap T_2$.
Since $zw\notin E(G)$, by (2) of Lemma~\ref{lem:relationt1t2},
$z$ is not adjacent to $y$.
Since $y$ and $z$ are neighbors of $v$ and $y,z$ are not adjacent, by (2) of Lemma~\ref{lem:shortdistance}, there is an induced path $zpy$ for some $p\in S$.
Then by (3) of Lemma~\ref{lem:relationt1t2}, $p$ is adjacent to $w$, and therefore, $p$ is a bypassing vertex.

Suppose $z\in T_1\cap T_2$. Recall that $x$ is a vertex in $(T_1\setminus T_2)\cap S$. If $x$ is adjacent to $y$, then by (2) of Lemma~\ref{lem:relationt1t2},
$x$ is adjacent to both $w$ and $z$, and thus $x$ is a bypassing vertex.
We may assume that $xy\notin E(G)$.
Then by (2) of Lemma~\ref{lem:shortdistance}, there is an induced path $xpy$ for some $p\in S$.
By (3) of Lemma~\ref{lem:relationt1t2}, $p$ is adjacent to both $w$ and $z$, and therefore, $p$ is a bypassing vertex, as required.

  \begin{figure}[t]
\begin{subfigure}[b]{0.5\textwidth}
      \includegraphics[scale=0.5]{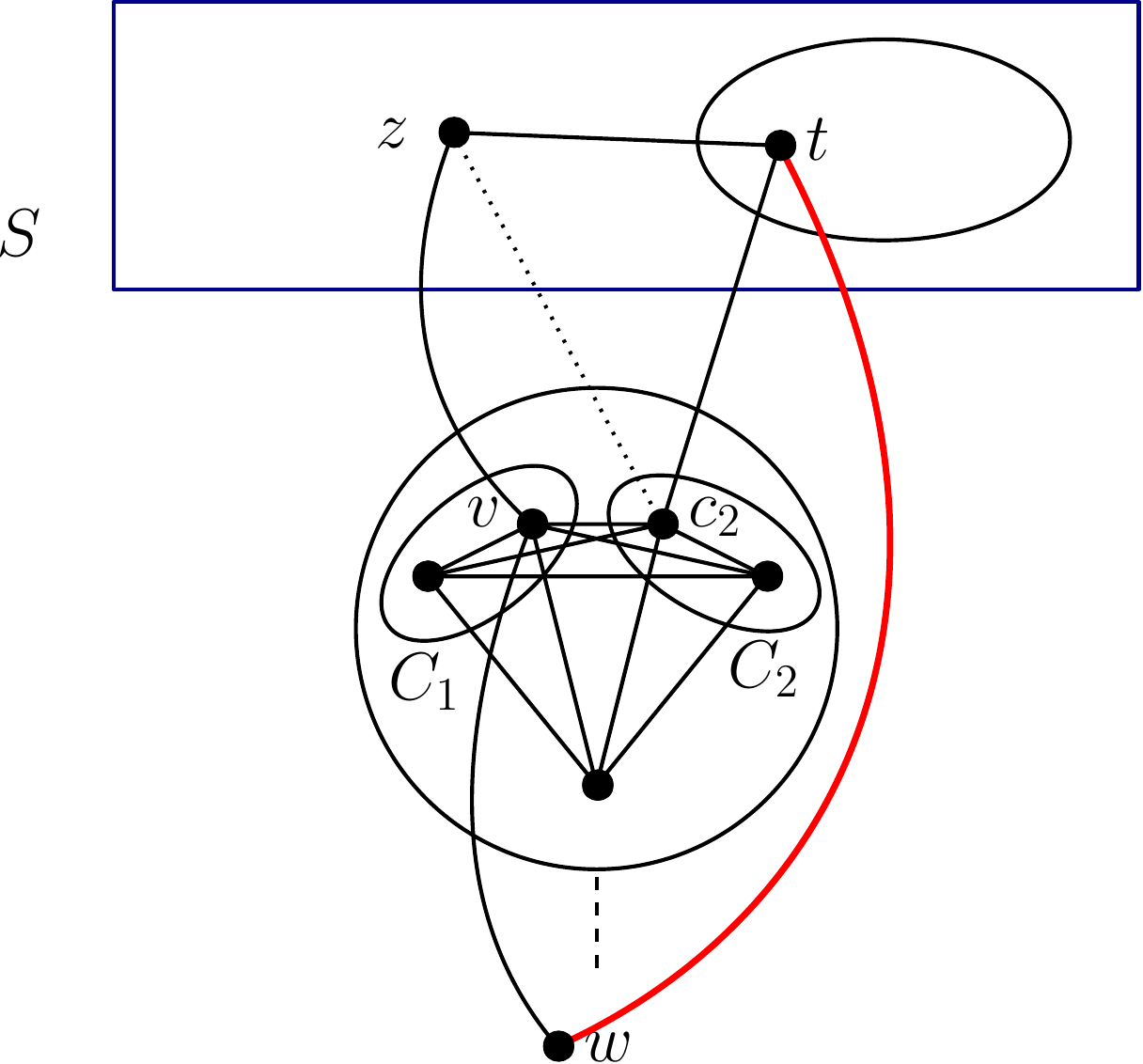}
       \caption{$z\in (T_1\setminus T_2)\cap S$.}
\end{subfigure}
\qquad
\begin{subfigure}[b]{0.5\textwidth}
      \includegraphics[scale=0.5]{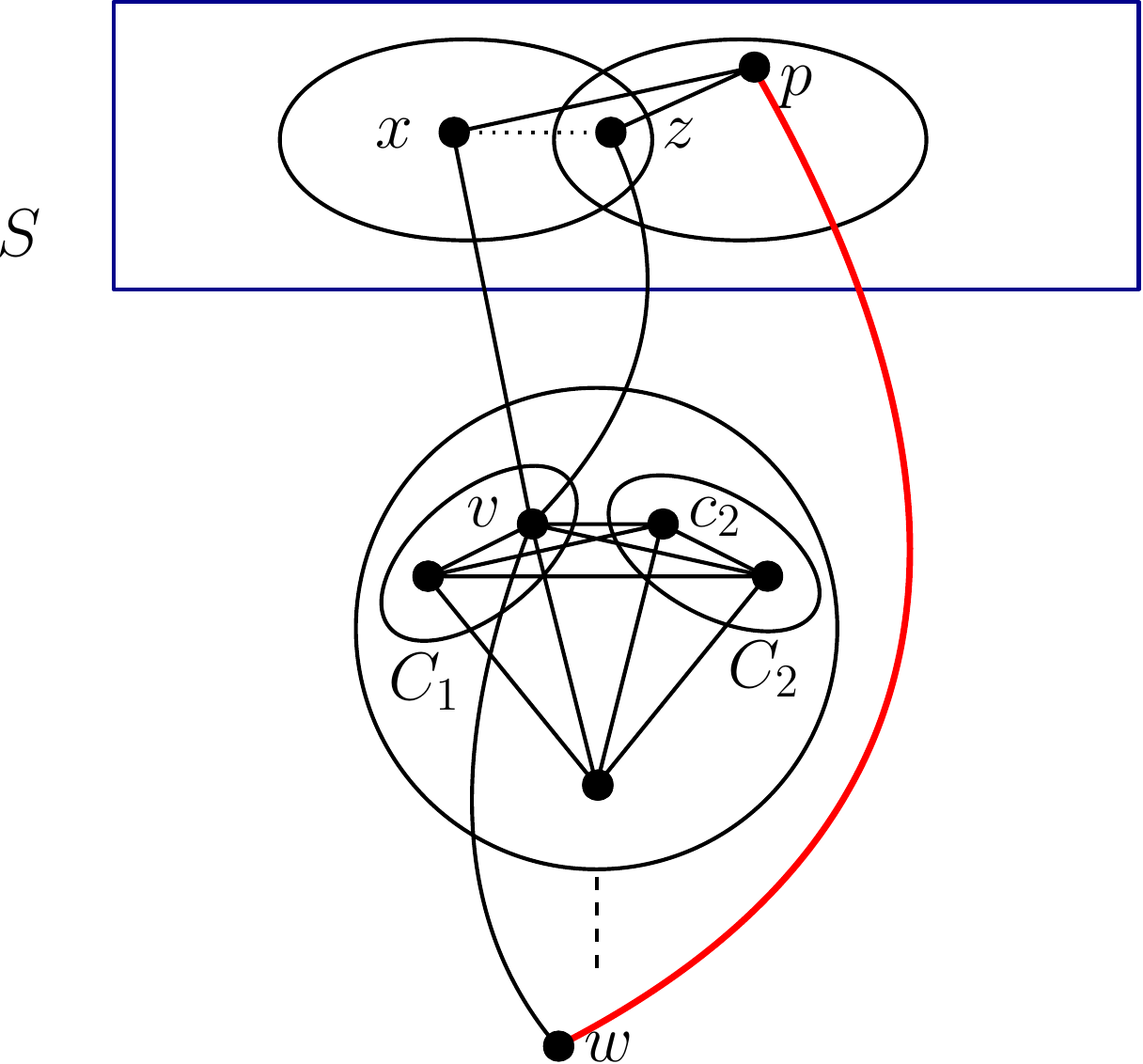}
       \caption{$z\in (T_1\cap T_2)\cap S$.}
\end{subfigure}
\center
\begin{subfigure}[b]{0.5\textwidth}
      \includegraphics[scale=0.5]{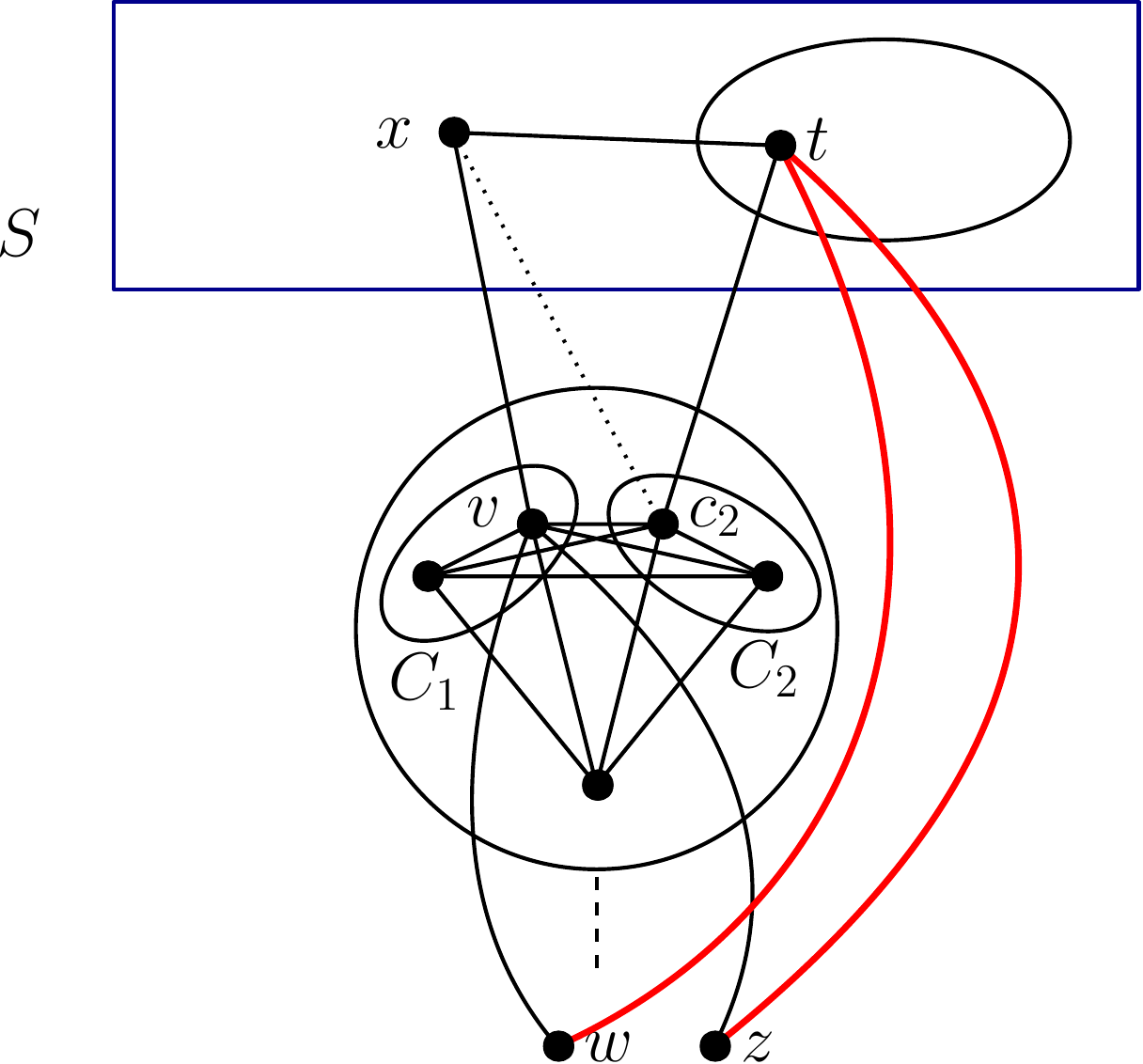}
       \caption{$z\in (T_1\cap T_2)\cap V(G-S)$.}
\end{subfigure}
            \caption{When $B$ is a complete bag and $w\in (T_1\cap T_2)\cap V(G-S)$ in Lemma~\ref{lem:twinclassreduction}.
            The red thick edges illustrate the edges whose existence is guaranteed by Lemma~\ref{lem:completerelationt1t2} or non-existence of small DH obstructions. } \label{fig:lemma45star2}
  \end{figure}

\subparagraph{\textbf{Case 2.} $B$ is a complete bag:}  
Note that $C_1$ is complete to $C_2$, and  $w$ is contained in either  $C_2$ or $(T_1\cap T_2)\cap V(G-S)$.
We first discuss when $w$ is contained in $C_2$.

Suppose $w\in C_2$. As $w$ is not adjacent to $z$, $z$ cannot be in $T_1\cap T_2$, and furthermore $z$ cannot be in $B$ as $B$ is a complete bag. 
Thus $z\in (T_1\setminus T_2)\cap S$. 
If $z$ has a neighbor in $T_2\cap S$, then the neighbor is a bypassing vertex, because it is adjacent to both $w$ and $z$.
We may assume that $z$ has no neighbors in $T_2\cap S$.
Observe that $z$ and $T_2\cap S$ are contained in the same connected component of $G[S]$, otherwise, Branching Rule~\ref{brule:reducecomponent} can be applied.
Let us take a shortest path $P$ from $z$ to $T_2\cap S$ in $G[S]$.
Then $Pw$ is an induced path of length at least $3$ and $v$ is adjacent to its end vertices, and thus $G[S\cup \{v,w\}]$ has a DH obstruction by Lemma~\ref{lem:dhobs}, which contradicts the assumption that $(G,S,k)$ is reduced under Branching Rule~\ref{brule:threevertices}.

Now, suppose $w\in (T_1\cap T_2)\cap V(G-S)$. In this case, $z$ is not in $C_2$. We distinguish subcases by the places of $z$.
We illustrate cases in Figure~\ref{fig:lemma45star2}.

\subparagraph{\textbf{Case 2-1.} $z\in (T_1\setminus T_2)\cap S$ :}  
Let $P$ be a shortest path from $z$ to $T_2\cap S$.
If $P$ has length at least $2$, then $Pc_2$ is an induced path of length at least $3$ and $v$ is adjacent to its end vertices.
So, $G[S\cup \{v,c_2\}]$ contains a DH obstruction, which is a contradiction.
Thus, $z$ has a neighbor in $T_2\cap S$, say $t$.
If $w$ is not adjacent to $t$, then $ztc_2w$ is an induced path and $v$ is adjacent to its end vertices.
This contradicts the assumption that $(G,S,k)$ is reduced under Branching Rule~\ref{brule:threevertices}.
Thus, $wt\in E(G)$ and $t$ is a bypassing vertex.

\subparagraph{\textbf{Case 2-2.} $z\in (T_1\cap T_2)\cap S$ :}  
Recall that $x$ is a vertex in $(T_1\setminus T_2)\cap S$.
If $x$ is adjacent to $z$, then by (1) of Lemma~\ref{lem:completerelationt1t2}, 
$w$ is adjacent to $x$
because $w$ is not adjacent to $z$. Thus, $x$ is a bypassing vertex.
So, we may assume that $x$ is not adjacent to $z$.
By (2) of Lemma~\ref{lem:shortdistance}, there is an induced path $xpz$ for some $p\in S$.

If $p\in (T_1\setminus T_2)\cap S$, then $p$ is adjacent to $w$ by (1) of Lemma~\ref{lem:completerelationt1t2}, and thus $p$ is a bypassing vertex.
Assume $p\in (T_1\cap T_2)\cap S$. If $p$ is adjacent to $w$, then $p$ is a bypassing vertex, and we are done.
Otherwise, by (1) of Lemma~\ref{lem:completerelationt1t2}, $x$ should be adjacent to both $w$ and $z$, since $wz\notin E(G)$.
Therefore, we may assume that $p\notin T_1$. 
Then 
by (2) of Lemma~\ref{lem:completerelationt1t2}, 
we have $p\in (T_2\setminus T_1)\cap S$, and again by (1) of Lemma~\ref{lem:completerelationt1t2},
either $wz\in E(G)$ or $wp\in E(G)$.
Since $wz\notin E(G)$, $p$ becomes a bypassing vertex.

\subparagraph{\textbf{Case 2-3.} $z\in (T_1\cap T_2)\cap V(G-S)$ :}  
If $x$ is adjacent to $w$ or $z$, then by (1) of Lemma~\ref{lem:completerelationt1t2},  
$x$ is adjacent to both $w$ and $z$, because $wz\notin E(G)$.
Then $x$ is a bypassing vertex. Therefore, we may assume $x$ is adjacent to neither $w$ nor $z$.
We take a shortest path $P$ from $x$ to $T_2\cap S$.
If $P$ has length at least $2$, then $Pc_2$ is an induced path of length at least $3$, and 
since $v$ is adjacent to its end vertices, $G[V(P)\cup \{v, c_2\}]$ contains a DH obstruction by Lemma~\ref{lem:dhobs}.
But this contradicts the assumption that $G$ is reduced under Branching Rule~\ref{brule:threevertices}.
We may assume that $P$ has length $1$, and let $t$ be a neighbor of $x$ in $T_2\cap S$.
Observe that if $t$ is not adjacent to $w$ or $z$, then $xtc_2w$ or $xtc_2z$ is an induced path, respectively, 
and $v$ is adjacent to its end vertices. It contradicts the assumption that $(G,S,k)$ is reduced under Branching Rule~\ref{brule:threevertices}.
Therefore $t$ is adjacent to both $w$ and $z$, which implies that $t$ is a bypassing vertex.

\medskip 
We conclude that, for every induced path $w'wvzz'$, there exists a bypassing vertex for $v$. It contradicts the assumption that $D$ is reduced under Reduction Rule~\ref{rrule:p5middle}.
Therefore, every bag of $D$ contains at most one $S$-attached twin class.
\end{proof}

Next, we show that for a bag $B$ of $D$, if a connected component of $D-V(B)$ contains no $S$-attached bags, then $B$ is a simple star bag adjacent to the component.

\begin{lemma}
\label{lem:smallbranch}
Let $B$ be a bag and $D_1$ be a connected component of $D-V(B)$ containing no $S$-attached bags.
Then $B$ is a simple star bag whose center is adjacent to $D_1$.
\end{lemma}
\begin{proof}
Let $B_1$ be the neighbor bag of $B$ contained in $D_1$.
First claim that $D_1=B_1$. Suppose $D_1$ contains at least one bag other than $B_1$.
We regard $B_1$ as the root bag of $D_1$, and choose a bag $Y$ in $D_1$ with maximum $\abs{\btwn (Y, B_1)}$. 
Clearly, $Y$ is a leaf bag in $D$.
Let $X$ be the neighbor bag of $Y$.

Suppose $X$ is a star.
As we choose $Y$ with maximum $\abs{\btwn (Y, B_1)}$, 
every child of $X$ is a leaf bag.
We claim that there is no leaf bag of $D$ pending to a leaf of $X$.
Suppose for contradiction there exists such a bag $Y_1$.
Since $D$ is canonical, $Y_1$ is not a star whose center is adjacent to $X$.
If $Y_1$ is a star whose leaf is adjacent to $X$, then it can be reduced under Reduction Rule~\ref{rrule:leaftoS}.
If $Y_1$ is a complete graph, then it can be reduced under Reduction Rule~\ref{rrule:leaf}, 
which is a contradiction. 
Therefore, there is no leaf bag of $D$ pending to a leaf of $X$, 
and it implies that 
the center of $X$ is adjacent to $Y$, and
$Y$ is the unique child of $X$.

Let $v$ be an unmarked vertex of $X$. As $N_G(v)$ is the set of unmarked vertices in $Y$ which is a twin class, 
any induced path of length $4$ could not contain $v$ as the third vertex.
Therefore, Reduction Rule~\ref{rrule:p5middle} can be applied to remove $v$, which is a contradiction.

Suppose $X$ is a complete graph.
Since $D$ is canonical, $Y$ is not a complete graph.
If $Y$ is a star whose leaf is adjacent to $X$, 
then all unmarked leaves in $Y$ can be removed by Reduction Rule~\ref{rrule:leaftoS}.
If $Y$ is a star whose center is adjacent to $X$, then 
we can apply Reduction Rule~\ref{rrule:leaf} to $X$ and $Y$.

We conclude that $D_1=B_1$. Moreover, if $B$ is not a star whose center is adjacent to $B_1$, then we can reduce $B_1$ using Reduction Rule~\ref{rrule:leaftoS} or \ref{rrule:leaf}.
Thus $B$ is a star whose center is adjacent to $B_1$.
\end{proof}

The following structure is illustrated in Figure~\ref{fig:rrule2}.

  \begin{figure}[t]
      \centering
      \includegraphics[scale=0.55]{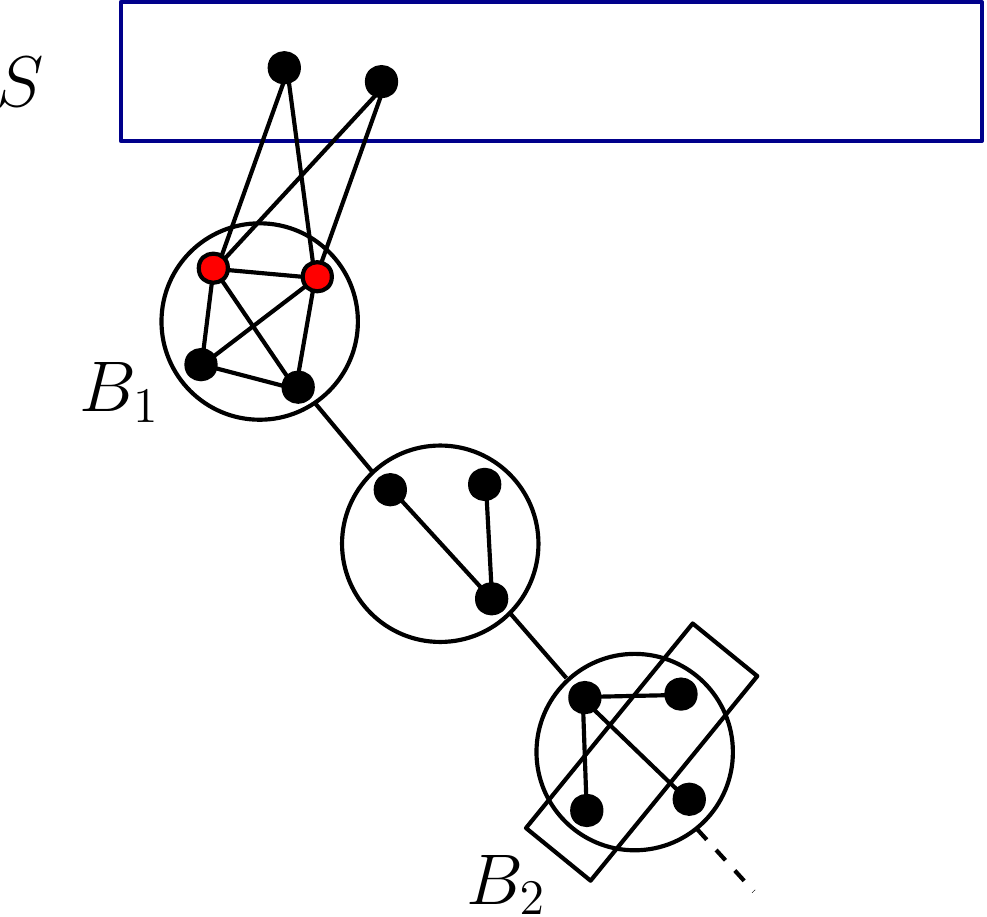}
            \caption{Lemma~\ref{lem:starcenter}.} \label{fig:rrule2}
  \end{figure}
  
\begin{lemma}
  \label{lem:starcenter}
Let $B_1$ be a leaf bag containing at most  one $S$-attached twin class and $B_2$ be a bag distinct from $B_1$ such that
\begin{itemize}
\item $B_2$ is a star bag whose center is adjacent to $\comp (B_2, B_1)$.
\item every bag in $\btwn (B_1, B_2)\setminus \{B_1, B_2\}$ is not a $(B_1, B_2)$-separator bag, and has exactly two neighbor bags, and 
\item for every bag $B$ in $\btwn (B_1, B_2)\setminus \{B_1, B_2\}$ that is not a star whose center is adjacent to $\comp (B, B_1)$, 
$B$ is non-$S$-attached.
\end{itemize}
Then $B_2$ contains no non-$S$-attached twin class $C$.
\end{lemma}
\begin{proof}
Suppose $B_2$ contains a non-$S$-attached twin class, and 
let $v$ be a vertex in the class.
We claim that there is no induced path $w'wvzz'$, which implies that Reduction Rule~\ref{rrule:p5middle} can be applied.
Suppose there is such a path.

We claim that either $N_G(w)\subseteq N_G(z)$ or $N_G(z)\subseteq N_G(w)$. If $N_G(w)\subseteq N_G(z)$ then $w'$ should be adjacent to $z$, which contradicts the fact that 
$w'wvzz'$ is an induced path. The same argument holds when $N_G(z)\subseteq N_G(w)$.

Let $P_w$ and $P_z$ be the bags containing $w$ and $z$, respectively. As $B_2$ is a star whose center is adjacent to $\comp (B_2, B_1)$,
$P_w$ and $P_z$ are bags in $\btwn (B_1, B_2)\setminus \{B_2\}$. 

First assume $P_w=P_z=B_1$. In this case, since $w$ is not adjacent to $z$, $B_1$ is a star bag.
Note that no DH obstruction contains two twins, and therefore, $w$ and $z$ are contained in distinct twin classes.
Since $B_1$ contains at most one $S$-attached twin class by Lemma~\ref{lem:twinclassreduction}, 
one of $w$ and $z$ is contained in the non-$S$-attached twin class in $B_1$.
Say $w$ is such a vertex. Then we have $N_G(w)\subseteq N_G(z)$, because $w$ and $z$ are twins in $G-S$.

Now, we assume at least one of $P_w$ and $P_z$ is not equal to $B_1$.
We further assume $P_z$ is contained in $\btwn (P_w, B_2)\setminus \{B_2\}$. The same argument holds when $P_w$ is contained in $\btwn (P_z, B_2)\setminus \{B_2\}$.

Since $P_z\in \btwn (P_w, B_2)\setminus \{B_2\}$ and $wz\notin E(G)$, $P_z$ is not a complete bag. 
Thus, $P_z$ is a star bag whose center is adjacent to $\comp (P_z, B_2)$.
As $P_z\neq B_1$ and it is not $S$-attached by the assumption, all neighbors of $z$ in $G$ are neighbors of $w$.
Then $z'$ should be adjacent to $w$, which is a contradiction.

We conclude that there is no such path $w'wvzz'$, and Reduction Rule~\ref{rrule:p5middle} can be applied to remove $v$. Therefore, $B_2$ contains no non-$S$-attached twin class $C$.
\end{proof}

The following structure is illustrated in Figure~\ref{fig:rrule3}.
 
\begin{lemma}\label{rrule:sattachedleaf}
Let $B_1$ be a leaf bag having exactly one $S$-attached twin class and $B_2$ be a simple star bag distinct from $B_1$ such that 
\begin{itemize}
\item $B_1$ is not a star whose leaf is adjacent to a neighbor bag,
\item every bag in $\btwn (B_1, B_2)\setminus \{B_1, B_2\}$ is non-$S$-attached, not a $(B_1, B_2)$-separator bag and has exactly two neighbor bags.
\end{itemize}
Then $B_1$ contains no non-$S$-attached twin class.
\end{lemma}
\begin{proof}
We claim that if $B_1$ contains a non-$S$-attached twin class, then we can apply a reduction rule to remove it.
Suppose $B_1$ contains a non-$S$-attached twin class $C_1$, 
and let $C_2$ be the $S$-attached twin class in $B_1$.

Let $v\in C_1$ and we claim that there is no induced path $w'wvzz'$.
If this is true, then we can apply Reduction Rule~\ref{rrule:p5middle}. 
Suppose there is such an induced path.

Assume that $w\in V(B_1)$. In this case, $B_1$ should be a complete bag.
Therefore, $z$ is adjacent to $w$, 
because $z\in V(G-S)$, and $w$, $v$ are twins in $G-S$.
This contradicts the assumption that $w'wvzz'$ is an induced path. Thus, we can assume that $w\notin V(B_1)$, and similarly, $z\notin V(B_1)$.

By symmetry, we assume $\abs{\btwn (B_1, B_w)}\le \abs{\btwn (B_1, B_z)}$, where $B_w$ and $B_z$ are bags containing $w$ and $z$, respectively.
Since $w$ is not adjacent to $z$, $B_w$ should be a star bag whose center is adjacent to the component $\comp (B_w, B_1)$. 
Therefore, every neighbor of $w$ in $G-S$ is adjacent to $z$, and in particular, $w'$ is adjacent to $z$. 
This contradicts the assumption that $w'wvzz'$ is induced.
This proves the claim.
It contradicts the assumption that $D$ is reduced.
\end{proof}

  \begin{figure}[t]
    \centering
      \includegraphics[scale=0.5]{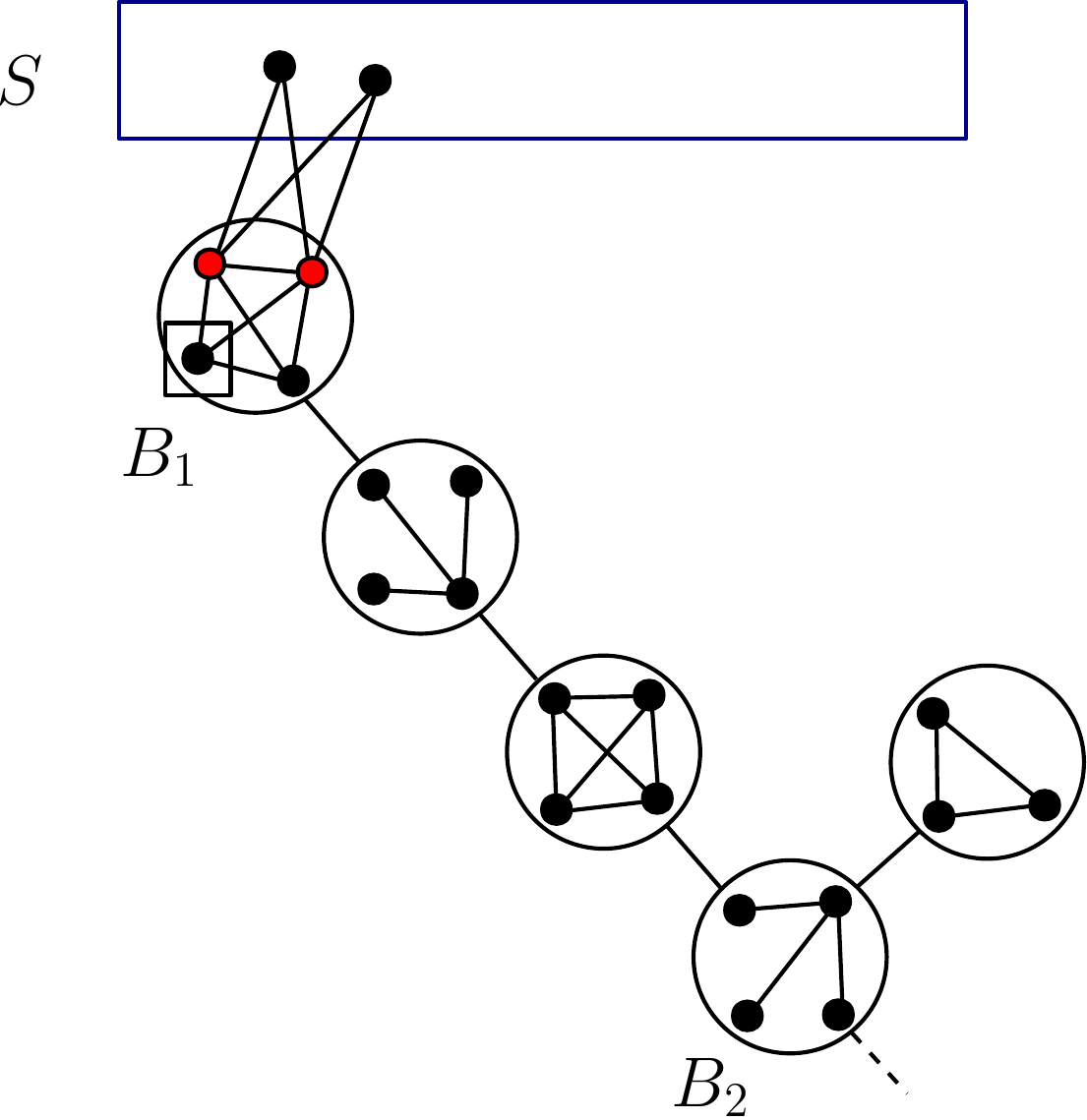}

    \caption{Lemma~\ref{rrule:sattachedleaf}.} \label{fig:rrule3}
      
  \end{figure}

\begin{lemma}
\label{lem:simplifynearsattached}
Let $B$ be a simple star bag, and let $D_1$ be a  connected component of $D-V(B)$ such that
\begin{itemize}
\item $D_1$ contains exactly one $S$-attached bag $B_1$, and
\item there is no $(B_1, B)$-separator bag.
\end{itemize}
Then $B_1$ is a star whose leaf is adjacent to $\comp (B_1, B)$ and there is a leaf bag $B_2$ where the center of $B_1$ is adjacent to $B_2$.
\end{lemma}
\begin{proof}
We first claim that $B_1$ is a star whose leaf is adjacent to $\comp (B_1, B)$. We prove this by a sequence of auxiliary claims.
Suppose for contradiction that $B_1$ does not satisfy the property; that is, either $B_1$ is a complete bag or a star bag whose center is adjacent to $\comp (B_1, B)$.

\begin{claim}
There is no connected component of $D-V(B_1)$ other than $\comp (B_1, B)$.
\end{claim}
\begin{clproof}
If there is such a component $C_1$, then by the assumption, it contains no $S$-attached bag.
By Lemma~\ref{lem:smallbranch}, $B_1$ is a star whose center is adjacent to $C_1$, 
contradicting our assumption. Thus, there is no connected component of $D-V(B_1)$ other than $\comp (B_1, B)$. 
\end{clproof}

We observe that $B_1$ contains one $S$-attached twin class by Lemma~\ref{lem:twinclassreduction}.
Also, all bags in $\btwn (B, B_1)\setminus \{B, B_1\}$ have exactly two neighbor bags.
This follows from Lemma~\ref{lem:smallbranch} and the fact that every bag in $\btwn (B_1, B)\setminus \{B, B_1\}$ is not a $(B, B_1)$-separator bag. 
Now, we can observe that $B$ and $B_1$ satisfy the conditions of Lemma~\ref{rrule:sattachedleaf}.
Therefore, $B_1$ contains no non-$S$-attached twin class.

\begin{claim}\label{claim:starreverse}
There is no star bag $B_2\in \btwn (B, B_1)\setminus \{B, B_1\}$ whose center is adjacent to $\comp (B_2, B_1)$.
\end{claim}
\begin{clproof}
Suppose there is such a star bag $B_2$.
Then $B_1$ and $B_2$ satisfy conditions in Lemma~\ref{lem:starcenter}.
Thus, $B_2$ has no non-$S$-attached twin class.
But this is impossible as $B_2$ has only two neighbor bags and $B_2$ has no $S$-attached twin class.
Thus, such a bag $B_2$ does not exist.
\end{clproof}

By Claim~\ref{claim:starreverse}, we observe that $B_1$ and its parent bag satisfy the condition (1) or (2) of Reduction Rule~\ref{rrule:leaf}, and thus we can apply the rule.
It contradicts the assumption that $D$ is reduced.
Thus, $B_1$ is a star whose leaf is adjacent to $\comp (B_1, B)$.

Now, suppose there is no bag $B_2$ where the center of $B_1$ is adjacent to $B_2$.
Since there is no bag pending to a leaf of $B_1$ by Lemma~\ref{lem:smallbranch},
$B_1$ is a leaf bag. In this case, we can reduce using Reduction Rule~\ref{rrule:leaftoS}, which is a contradiction.
Therefore, there is a leaf bag $B_2$ where the center of $B_1$ is adjacent to $B_2$, as required.
\end{proof}

The next lemma describes the structure of $\btwn(B_1, B_2)$ where $B_1$ and $B_2$ are simple star bags, and there is no $S$-attached bag in the connected component of $D-V(B_1)-V(B_2)$ containing bags in $\btwn (B_1, B_2)\setminus \{B_1, B_2\}$.

\begin{lemma}\label{rrule:bypassing1}
Let $B_1$ and $B_2$ be two simple star bags in $D$ such that 
\begin{itemize}
\item every bag in $\btwn (B_1, B_2)\setminus \{B_1,B_2\}$ is a non-$S$-attached bag, has two neighbor bags, and is not a $(B_1, B_2)$-separator bag. 
\end{itemize}
Then $B_1$ and $B_2$ are neighbor bags.
\end{lemma}
\begin{proof}
Suppose for contradiction that $\btwn (B_1, B_2)\setminus \{B_1,B_2\}\neq \emptyset$.
Let $B\in \btwn (B_1, B_2)\setminus \{B_1,B_2\}$ and $v$ be an unmarked vertex of $B$.

We claim that there is no induced path $w'wvzz'$.
Suppose there is such an induced path.
By symmetry, we assume $\abs{\btwn (B_1, B_w)}\le \abs{\btwn (B_1, B_z)}$, where $B_w$ and $B_z$ are bags containing $w$ and $z$, respectively.
If $w$ and $z$ are contained in the different connected components of $D-V(B)$, 
then because $B$ is the not $(B_1, B_2)$-separator bag, $B$ should be a complete bag.
But then $w$ is adjacent to $z$, contradiction.
Thus, $w$ and $z$ are contained in the same connected component of $D-V(B)$.
Without loss of generality, we may assume that such a connected component contains $B_1$.

Suppose there is a bag $B_1'$ where the center of $B_1$ is adjacent to $B_1'$.
Since $B_1$ is simple, $B_1'$ contains only non-$S$-attached twin class.
Thus one of $w$ and $z$ are not contained in $B_1'$, as they are not twins in the path $w'wvzz'$.

We may assume $w$ is in $\btwn (B_1, B)\setminus \{B_1, B\}$.
Then, every neighbor of $w$ in $G-S$ is adjacent to $z$, in particular, $w'$ is adjacent to $z$. 
This contradicts the assumption that $w'wvzz'$ is induced.

This proves the claim. Since there is no such a path $w'wvzz'$, we can apply Reduction Rule~\ref{rrule:p5middle} to remove $v$.
It contradicts the assumption that $D$ is reduced.
We conclude that $B_1$ and $B_2$ are neighbor bags.
\end{proof}

Finally, we claim that our instance has the desired inseparability property. We formalize and prove this property below.

\begin{lemma}
\label{lem:simplifynearsattached2}
Let $B$ be a bag and let $D_1$ be a connected component of $D-V(B)$ such that
$D_1$ contains exactly one $S$-attached bag $B_1$.
Then there is no $(B_1, B)$-separator bag.
\end{lemma}
\begin{proof}
For contradiction, suppose that there is a $(B_1, B)$-separator bag. We choose such a bag  $B_2$ with minimum $\abs{\btwn (B_1, B_2)}$.
From the choice of $B_2$, there is no $(B_1, B_2)$-separator bag. 

We verify preconditions of Lemma~\ref{lem:simplifynearsattached} for $B_1$ and $B_2$.
Clearly, $\comp (B_2, B_1)$ has exactly one $S$-attached bag $B_1$, 
and there is no $(B_1, B_2)$-separator bag. 
To see that $B_2$ is a simple star bag, let us assume that there is a connected component $D_2$ of $D-V(B_2)$ where the center of $B_2$ is adjacent to $D_2$; if there is no such a component, it is clear by definition.
As $D_2$ contains no $S$-attached bag, by Lemma~\ref{lem:smallbranch}, $D_2$ consists of one bag and $B_2$ is a simple star bag.

By applying Lemma~\ref{lem:simplifynearsattached} for $B_1$ and $B_2$, 
we can observe that $B_1$ is a star whose leaf is adjacent to $\comp (B_1, B_2)$, and 
there is a leaf bag $B_l$ where the center of $B_1$ is adjacent to $B_l$.
We can also observe that $B_1$ is a simple star bag.
By Lemma~\ref{lem:smallbranch}, there is no connected component of $D-V(B_1)$ pending to leaves of $B_1$ other than the leaf adjacent to its parent.

Note that every bag $A$ in $\btwn (B_1, B_2)\setminus \{B_1, B_2\}$ has two neighbor bags, because it is not a $(B_1, B_2)$-separator bag and by Lemma~\ref{lem:smallbranch} there is no other component $D-V(A)$ pending to $A$.
Therefore  by Lemma~\ref{rrule:bypassing1}, $B_2$ is a neighbor bag of $B_1$.

Now, by Lemma~\ref{lem:starcenter}, there is no non-$S$-attached twin class in $B_1$, 
which means that the unmarked vertices of $B_1$ form a unique $S$-attached twin class. 
Then, we can apply Reduction Rule~\ref{rrule:bypassing2} to $B_{l},B_1, B_2$, a contradiction.

We conclude that there is no $(B_1, B)$-separator bag.
\end{proof}

\subsection{Connected components with two $S$-attached bags}\label{sec:twinclass}

This section is devoted to showing that if $D$ is reduced and contains two distinct $S$-attached classes, then we can apply a reduction rule.
Suppose $D$ is reduced and contains two distinct $S$-attached classes, and we choose a root bag of $D$. 
Let $B$ be a farthest bag from the root bag 
 such that there are two descendant bags $B_1$ and $B_2$ of $B$ having distinct $S$-attached twin classes $C_1$ and $C_2$, respectively.

First, we verify that the distance from $C_1$ to $C_2$ in $G-S$ is at most $2$.

\begin{lemma}
\label{lem:distancec1c2}
The distance from $C_1$ to $C_2$ in $G-S$ is at most $2$.
\end{lemma}
\begin{proof}
Let us take a shortest sequence of twin classes $(C_1=U_0)-U_1- \cdots - U_t-(C_2=U_{t+1})$ from $C_1$ to $C_2$ in $G-S$ such that for $i,j\in \{0, 1, \ldots, t+1\}$ with $i\neq j$, $U_i$ is complete to $U_j$ if $\abs{i-j}=1$ and they are anti-complete, otherwise. We note that each $U_i$ except $U_0$ and $U_{t+1}$ corresponds to a $(B_1, B_2)$-separator bag.
 Clearly, at most one of $U_1, \ldots, U_t$ possibly has a neighbor in $S$ because $C_i$ is the unique $S$-attached twin class in $\comp (B, B_i)$ if $B_i\neq B$.
 By (3) and (4) of Lemma~\ref{lem:shortdistance}, 
 the length from $C_1$ to $C_2$ in $G-S$ cannot be $3$ or $4$, and thus $t$ cannot be $2$ or $3$.
 Also, by Lemma~\ref{lem:simplifynearsattached2}, 
 we know that there is no $(B_i, B)$-separator bag when $B_i\neq B$.
 Thus, $t$ cannot be larger than $3$.
So, the distance from $C_1$ to $C_2$ in $G-S$ is at most  $2$. 
\end{proof}

\begin{PROP}
  \label{prop:anticomplete1}
 The bag $B$ is not a $(C_1, C_2)$-separator bag.
  \end{PROP}
\begin{proof}
For each $i\in \{1,2\}$ let $T_i=N_G(C_i)$.
Since by Lemma~\ref{lem:distancec1c2} the distance from $C_1$ to $C_2$ is at most $2$, it follows from Observation~\ref{obser:separator} that there exists at most one $(C_1, C_2)$-separator bag. Suppose that $B$ is the $(C_1, C_2)$-separator bag.
Note that $B_i\neq B$ for some $i\in \{1,2\}$ because $B_1$ and $B_2$ are distinct. Without loss of generality, we assume that $B_1\neq B$.
We verify the proposition by a sequence of claims.

\begin{claim}
$B_1$ is not a star bag whose leaf is adjacent to $\comp (B_1, B)$.
\end{claim}
\begin{clproof}
Suppose $B_1$ is a star bag whose leaf is adjacent to $\comp (B_1, B)$. As $B_1$ is the unique $S$-attached bag in $\comp (B, B_1)$, 
by Lemma~\ref{lem:smallbranch}, there is no bag pending to a leaf of $B_1$.
Also, the center of $B_1$ is marked, otherwise, we can apply Reduction Rule~\ref{rrule:leaftoS}, and by Lemma~\ref{lem:smallbranch}, $B_1$ is a simple star bag.
Therefore, $C_1$ consists of leaves of $B_1$, and $B_1$ is a $(C_1, C_2)$-separator bag. But it contradicts the assumption that $B\neq B_1$ and $B$ is the only $(C_1, C_2)$-separator bag.
\end{clproof}

Note that $B_1$ is either a complete graph, or a star whose center is adjacent to $\comp (B_1, B)$.
We observe that $B_1$ contains a non-$S$-attached twin class.
\begin{claim}
$B_1$ contains a non-$S$-attached twin class.
\end{claim}
\begin{clproof}
Suppose for contradiction that $B_1$ contains no non-$S$-attached twin class, that is, $C_1$ is exactly the set of unmarked vertices of $B_1$. 
Let $B_3$ be the parent bag of $B_1$. 
If $B_3$ is not a star whose center is adjacent to $B_1$, then we can apply Reduction Rule~\ref{rrule:leaf}.
We may assume $B_3$ is a star whose center is adjacent to $B_1$. But in this case, $B\neq B_3$, 
and thus, $B_3$ has no $S$-attached twin classes. By Lemma~\ref{lem:smallbranch}, $B_3$ has exactly two neighbor bags, and by Lemma~\ref{lem:starcenter},  
it contains no non-$S$-attached twin class. But this is impossible.
We conclude that $B_1$ contains a non-$S$-attached twin class.
\end{clproof}

\begin{claim}
There is a vertex $x$ in $(T_1\setminus T_2)\cap V(G-S)$ contained in a complete bag such that $x$ has no neighbors in $S$.
\end{claim}
\begin{clproof}
If $B_1$ is a complete bag, then the non-$S$-attached twin class is contained in $(T_1\setminus T_2)\cap V(G-S)$.
Assume $B_1$ is a star. Since $B$ is a star whose leaf is adjacent to $\comp(B, B_1)$, 
there is at least one bag in $\btwn (B_1, B)\setminus \{B_1, B\}$. Moreover, there is no star bag $B'$ in $\btwn (B_1, B)\setminus \{B_1, B\}$ whose center is adjacent to $\comp (B', B_1)$ by Lemma~\ref{lem:starcenter}.
Therefore, there is at least one complete bag in $\btwn (B_1, B)\setminus \{B_1, B\}$, which contains a vertex in $(T_1\setminus T_2)\cap V(G-S)$.
We choose $x$ to be such a vertex. Then $x$ is a vertex in $(T_1\setminus T_2)\cap V(G-S)$ having no neighbors in $S$, and 
also contained in a complete bag.
\end{clproof}

Since $x$ is contained in a complete bag, $x$ has a neighbor in $(T_1\cap T_2)\cap V(G-S)$.
By (1) of Lemma~\ref{lem:relationt1t2}, we have $(T_1\cap T_2)\cap S\neq \emptyset$.
Since $x\in T_1\setminus T_2$ and $x$ is adjacent to a vertex in $T_1\cap T_2$, by (2) of Lemma~\ref{lem:relationt1t2}, $x$ should be adjacent to all vertices in $(T_1\cap T_2)\cap S$, which contradicts the fact that $x$ has no neighbors in $S$.
\end{proof}

The following lemma describes all possible cases.

\begin{lemma}\label{lem:twosattachedbags}
Let $B$ be a farthest bag from the root bag 
 such that there are two descendant bags $B_1$ and $B_2$ of $B$ having distinct $S$-attached twin classes $C_1$ and $C_2$, respectively.
Then $B_1\neq B_2$ and one of the following happens:
\begin{enumerate}
\item The distance from $C_1$ to $C_2$ in $G-S$ is $2$ and the unique $(C_1, C_2)$-separator bag is contained in $\comp (B, B_i)$ for some $i\in \{1,2\}$.
\item  $C_1$ is complete to $C_2$ and either
\begin{itemize}
\item $B$ is a star bag and $C_i$ is the set consisting of the center of $B$ for some $i\in \{1,2\}$, or
\item  $B$ is a star bag whose center is adjacent to $\comp (B, B_i)$ for some $i\in \{1,2\}$.
\end{itemize}
\item  $C_1$ is complete to $C_2$ and $B$ is a complete bag.
\end{enumerate}
\end{lemma}
\begin{proof}
By Lemma~\ref{lem:twinclassreduction}, each bag contains at most one $S$-attached twin class and it follows that $B_1\neq B_2$.
By Lemma~\ref{lem:distancec1c2}
the distance from $C_1$ to $C_2$ in $G-S$ is at most $2$.
Suppose that the distance from $C_1$ to $C_2$ in $G-S$ is $2$. 
Then, there is a unique $(C_1, C_2)$-separator bag in $D$.
By Proposition~\ref{prop:anticomplete1}, $B$ cannot be the $(C_1, C_2)$-separator bag.
Thus, the unique $(C_1, C_2)$-separator bag in contained in $\comp (B, B_i)$ for some $i\in \{1,2\}$.
If the distance from $C_1$ to $C_2$ is $1$, then $C_1$ is complete to $C_2$, 
and in this case if $B$ is a star, then its center either consists of one class $C_i$ or is adjacent to one of $\comp (B, B_1)$ and $\comp (B, B_2)$.
\end{proof}

We show that in each of three cases in Lemma~\ref{lem:twosattachedbags}, 
we can apply a reduction rule.

\begin{PROP}
  \label{prop:anticomplete2}
  Suppose the distance from $C_1$ to $C_2$ in $G-S$ is $2$ and the unique $(C_1, C_2)$-separator bag is contained in $\comp (B, B_2)$.
  Then for every induced path $P=p_1p_2p_3p_4p_5$ with $p_3\in C_2$, 
  there is a bypassing vertex for $P$ and $p_3$.
  \end{PROP}
\begin{proof}
Let $T_i=N_G(C_i)$ for each $i\in \{1,2\}$. 
 We start by proving the following claim.
\begin{claim}
The bag $B_2$ is the $(C_1, C_2)$-separator bag.
\end{claim}
\begin{clproof}
For a contradiction, suppose that $B_2$ is not the $(C_1, C_2)$-separator bag and let $B'\neq B_2$ be the $(C_1, C_2)$-separator bag.
Then $B'$ is a $(B_2, B)$-separator bag.
However, since $\comp (B, B_2)$ has exactly one $S$-attached bag $B_2$, 
by Lemma~\ref{lem:simplifynearsattached2}, 
there is no $(B_2, B)$-separator bag, which is a contradiction.
\end{clproof}

By Lemma~\ref{lem:smallbranch}, $B_2$ has no child pending to a leaf of $B_2$.
If the center of $B_2$ is unmarked, then we can reduce it using Reduction Rule~\ref{rrule:leaftoS}.
Thus there is component attached to the center of $B_2$, and by Lemma~\ref{lem:smallbranch} this component is a single leaf bag.
We call the leaf bag $B_3$, and let $C_3$ be the set of unmarked vertices of $B_3$. Note that $C_3$ is a non-$S$-attached twin class.
Also, by Lemma~\ref{lem:starcenter}, $B_2$ contains no non-$S$-attached twin class.

Suppose  there is an induced path $P=p_1p_2p_3p_4p_5$ with $p_3\in C_2$.
We want to show that there is a bypassing vertex for $P$ and $p_3$.
Observe that every neighbor of $p_3$ in $G$ is either in $S$ or in $C_3$.
As $C_3$ is a twin class, it contains at most one of $p_2$ and $p_4$.
If $p_2$ and $p_4$ are contained in $S$, then by Lemma~\ref{lem:twovertexinS}, there is a bypassing vertex for $p_3$.
Thus, we may assume that one of $p_2$ and $p_4$ is contained in $S$ and the other is contained in $C_3$.

By symmetry, we may assume that $p_2\in C_3$ and $p_4\in S$.
Note that since $p_2\in C_3$, $p_2$ has no neighbors in $S$.
Furthermore, 
as the distance from $C_1$ to $C_2$ is exactly $2$, $C_3$ is complete to $C_1$.
It implies that $p_2\in (T_1\cap T_2)\cap V(G-S)$.

By (1) of Lemma~\ref{lem:relationt1t2}, we have $(T_1\cap T_2)\cap S\neq \emptyset$.
Let $t\in (T_1\cap T_2)\cap S$.
We divide cases depending on whether $p_4$ is in $T_2\setminus T_1$ or $T_1\cap T_2$.

\subparagraph{\textbf{Case 1.} $p_4\in (T_2\setminus T_1)\cap S$:}  

Note that $p_4\in T_2\setminus T_1$ and $p_2, t\in T_1\cap T_2$.
Since $p_4$ is not adjacent to $p_2$, by (2) of Lemma~\ref{lem:relationt1t2},
$p_4$ is not adjacent to $t$ as well. As $p_4$ and $t$ are neighbors of $p_3$ and $(G,S, k)$ is reduced under Branching Rule~\ref{brule:reducecomponent}, 
$p_4$ and $t$ are contained in the same connected component of $G[S]$. Moreover, since $(G,S,k)$ is reduced under Branching Rule~\ref{brule:threevertices}, 
there is no induced path of length at least $3$ from $p_4$ to $t$ in $G[S]$, 
and thus the distance from $p_4$ to $t$ in $G[S]$ is exactly $2$.
Let $p_4pt$ be an induced path for some $p\in S$. 

If $p$ is contained in $T_1\cap T_2$, then $p_4$ should be adjacent to $p_2$ by (2) of Lemma~\ref{lem:relationt1t2}.
Thus,  $p$ is not contained in $T_1\cap T_2$.
If $p\in T_2\setminus T_1$, then by (2) of Lemma~\ref{lem:relationt1t2}, $p$ is adjacent to $p_2$, but $p_2$ has no neighbors in $S$, a contradiction. %
Lastly, assume that $p\in S\setminus T_2$. In this case,  by (3) of Lemma~\ref{lem:relationt1t2} with $(p,x,y_1,y_2)=(p,p_4,t,p_2)$, 
$p$ is adjacent to $p_2$, again a contradiction. %

\subparagraph{\textbf{Case 2.} $p_4\in (T_1\cap T_2)\cap S$:}  
Let $c_1\in C_1$.
If $p_2$ and $c_1$ have a common neighbor $c$ in $G-S$ that is adjacent to neither $p_3$ nor $p_4$, then
$G[\{p_2,p_3, p_4,c,c_1\}]$ is isomorphic to the house.
So, there are no such vertices. 
This implies that for each $i\in \{1,2\}$, there is no complete bag in $\btwn (B_i, B)\setminus \{B, B_i\}$,
and if $B_1$ or $B$ is a complete bag, then it contains no non-$S$-attached twin class.

We claim that $\btwn (B, B_1)$ contains at most $3$ bags.

\begin{claim}\label{claim:restricted}
$\btwn (B, B_1)$ contains at most $3$ bags, and when it contains $3$ bags, the bag in $\btwn (B, B_1)$ is a star bag whose center is adjacent to $B$.
\end{claim}
\begin{clproof}
Suppose $\btwn (B, B_1)$ contains more than $3$ bags, and let $B_1'$ be the parent bag of $B_1$ and $B_1''$ be the parent of $B_1'$.
As  $\btwn (B, B_1)\setminus \{B\}$ contains no complete bags,
both $B_1'$ and $B_1''$ are star bags. 
Thus, $B_1''$ is a star bag whose center is adjacent to $B_1'$.
Such a bag $B_1''$ does not exist by Lemma~\ref{lem:starcenter}. It proves the claim.
\end{clproof}

In particular, Claim~\ref{claim:restricted} implies that every neighbor of $C_3$ is either in $C_2$ or not contained in the component of $D-V(B)$ containing $B_2$.

We divide into subcases depending on the shape of $B_1$.

\subparagraph{\textbf{Case 2-1.} $B_1$ is a complete bag:}

First assume that $B_1=B$. As $B_1$ contains no non-$S$-attached twin class and $p_1p_4\notin E(G)$,
$p_1$ is in the neighborhood of $C_1$ in $G-S$. Then $c_1$ is adjacent to end vertices of $p_1p_2p_3p_4$, and by Lemma~\ref{lem:dhobs}, 
$G$ contains a small DH obstruction. This is a contradiction.
We may assume $B_1\neq B$.

As $D$ is canonical, 
the parent bag $B_1'$ of $B_1$ is a star bag.
We claim that $B=B_1'$. Suppose $B\neq B_1'$, that is, $B_1'$ is contained in $\btwn (B, B_1)\setminus \{B, B_1\}$.
Since there is no $(B, B_1)$-separator bag, the center of $B_1'$ is adjacent to either $\comp(B_1', B)$ or $B_1$.
As $B_1$ is the unique $S$-attached bag in $\comp(B, B_1)$, 
by Lemma~\ref{lem:smallbranch}, $B_1'$ has exactly two neighbor bags.
Also, again by Lemma~\ref{lem:smallbranch}, $B_1$ is a leaf bag.
Therefore, by Lemma~\ref{lem:starcenter}, the center of $B_1'$ is not adjacent to $B_1$.
On the other hand, if the center of $B_1'$ is adjacent to $\comp(B_1', B_1)$, 
then we can apply Reduction Rule~\ref{rrule:leaf}, as $B_1$ contains no non-$S$-attached twin class, which is a contradiction.
Thus, we have $B=B_1'$, and the same argument using Reduction Rule~\ref{rrule:leaf} implies that the center of $B_1'$ is adjacent to $B$.
Then $p_1$ should be contained in $B_1$ and adjacent to $p_4$, which is impossible.

\subparagraph{\textbf{Case 2-2.} $B_1$ is a star bag:}
First assume that $B_1=B$. 
As $C_1$ is complete to $C_3$, the center of $B_1$ is adjacent to $\comp (B, B_2)$.
If $B$ and $B_2$ are neighbor bags, then the marked edge connecting them can be recomposable.
Thus, in this case, $\btwn (B, B_2)$ contains $3$ bags. Let $B_4$ be the bag in $\btwn (B, B_2)\setminus \{B, B_2\}$, and $b$ be an unmarked vertex in $B_4$.
It is not difficult to observe that $p_1$ should be adjacent to $b$, since $p_1$ cannot be in $C_2$.
Then $bp_1p_2p_3p_4$ is an induced path and $c_1$ is adjacent to its end vertices, and thus $G$ contains a small DH obstruction.
It is a contradiction. We may assume $B_1\neq B$.

Similar to the case when $B_1$ is a complete bag, we can show that the parent of $B_1$ is $B$ and $B$ is a star whose center is adjacent to $B_1$.
In this case, $p_1$ is contained in the non-$S$-attached twin class, as $p_1$ is not adjacent to $p_4$.
As $B$ has at least $3$ vertices, there is a vertex $x\in V(G-S)$ where $x$ is adjacent to $p_1$, but not adjacent to $p_2, p_3, p_4$.
If $x$ is adjacent to $p_4$, then we have a small DH obstruction.
Otherwise, $xp_1p_2p_3p_4$ is an induced path, and $c$ is adjacent to its end vertices. 
It contradicts the non-existence of a small DH obstruction.

\medskip
We conclude that for every induced path $P=p_1p_2p_3p_4p_5$ with $p_3\in C_2$, 
  there is a bypassing vertex for $P$ and $p_3$.
\end{proof}

Next, we deal with the case when $C_1$ is complete to $C_2$. We prove the case when $B$ is a star bag in Proposition~\ref{prop:complete2}, 
and the case when $B$ is a complete bag in Proposition~\ref{prop:complete1}.

\begin{PROP}
  \label{prop:complete2}
  Suppose $C_1$ is complete to $C_2$, and either 
 \begin{itemize}
 \item $B_2=B$ and $C_2$ consists of the center of $B_2$ or
 \item  $B\neq B_2$, and $B$ is a star bag   whose center is adjacent to $\comp (B, B_2)$.
 \end{itemize}
  Then for every induced path $P=p_1p_2p_3p_4p_5$ with $p_3\in C_1$, 
  there is a bypassing vertex for $P$ and $p_3$.  
  \end{PROP}
\begin{proof}
For each $i\in \{1,2\}$ and let $c_i\in C_i$ and $T_i=N_G(C_i)\setminus (C_1\cup C_2)$.
We first observe that there is no child bag $B_1'$ pending to $B_1$ except the possible child in $\btwn (B, B_2)$ when $B=B_1$.
Suppose there is such a bag, and let $D'$ be the connected component of $D-V(B_1)$ containing $B_1'$.
As $D'$ has no $S$-attached bags by the choice of $B, B_1, B_2$, by Lemma~\ref{lem:smallbranch}, $B_1$ is a star whose center is adjacent to $D'$. 
Then $B_1$ becomes a $(C_1, C_2)$-separator bag, contradicting the assumption that $C_1$ is complete to $C_2$.
We conclude the claim.

As $B_i$ is the unique $S$-attached bag in $\comp (B, B_i)$ when $B\neq B_i$, 
by Lemma~\ref{lem:smallbranch}, 
every bag in $\btwn (B, B_i)\setminus \{B, B_i\}$ has exactly two neighbor bags.
Since either 
\begin{itemize}
 \item $B_2=B$ and $C_2$ consists of the center of $B_2$ or
\item $B$ is a star bag whose center is adjacent to $\comp (B, B_2)$,
\end{itemize}
every neighbor of a vertex in $C_1$ is contained in $\comp(B, B_1)$ or $\comp (B, B_2)$.

Suppose there is an induced path $P=p_1p_2p_3p_4p_5$ with $p_3\in C_1$. We will show that there is a bypassing vertex for $p_3$.
If $p_2, p_4\in S$, then it follows from Lemma~\ref{lem:twovertexinS}. 
Without loss of generality, we assume $p_2\in V(G-S)$.
We distinguish cases depending on whether $(T_1\cap T_2)\cap V(G-S)=\emptyset$ or not.

\subparagraph{\textbf{Case 1.} $(T_1\cap T_2)\cap V(G-S)\neq \emptyset$ :}  

We choose a vertex $x\in (T_1\cap T_2)\cap V(G-S)$.
Since every neighbor of a vertex in $C_1$ is contained in $\comp(B, B_1)$ or $\comp (B, B_2)$, 
$x$ is contained in $\comp (B, B_1)$ or $\comp (B, B_2)$.
Since $x$ is not contained in $C_1\cup C_2$, $x$ has no neighbors in $S$.

We claim that $(T_1\cap T_2)\cap S\neq \emptyset$.

\begin{claim}
$(T_1\cap T_2)\cap S\neq \emptyset$.
\end{claim}
\begin{clproof}
Suppose $T_1\cap S$ and $T_2\cap S$ are disjoint.
As $C_1$ is complete to $C_2$ and $(G,S, k)$ is reduced under Branching Rule~\ref{brule:reducecomponent}, $T_1\cap S$ and $T_2\cap S$ are contained in the same conncected component of $G[S]$.
Let $P$ be a shortest path from $T_1\cap S$ to $T_2\cap S$ in $G[S]$.
Since $Pc_2x$ is an induced path of length at least $3$ and $c_1$ is adjacent to its end vertices, 
$G[V(P)\cup \{c_1,c_2,x\}]$ contains a DH obstruction, which contradicts the assumption that $G$ is reduced under Branching Rule~\ref{brule:threevertices}.
Therefore, $(T_1\cap T_2)\cap S\neq \emptyset$. 
\end{clproof}

Let $y\in (T_1\cap T_2)\cap S$.
Clearly, $x$ is not adjacent to $y$.
Next, we claim that $((T_1\setminus T_2)\cup (T_2\setminus T_1))\cap S=\emptyset$.

\begin{claim}
$((T_1\setminus T_2)\cup (T_2\setminus T_1))\cap S=\emptyset$.
\end{claim}
\begin{clproof}
Suppose there is a vertex $u$ in $((T_1\setminus T_2)\cup (T_2\setminus T_1))\cap S$. 
Since $xy\notin E(G)$, if $uy\in E(G)$, then $ux\in E(G)$ by (1) of Lemma~\ref{lem:completerelationt1t2}. 
But $x$ has no neighbors in $S$. Thus, we have $uy\notin E(G)$.
Since $\{u,y\}\subseteq T_1\cap S$ or $\{u,y\}\subseteq T_2\cap S$, there is an induced path $upy$ for some $p\in S$.
We assume $\{u,y\}\subseteq T_1\cap S$; the symmetric argument holds when $\{u,y\}\subseteq T_2\cap S$.
If $p\in T_1\cup T_2$, then by (1) of Lemma~\ref{lem:completerelationt1t2}, $u$ or $p$ should be adjacent to $x$, which is a contradiction.
On the other hand, by (2) of Lemma~\ref{lem:completerelationt1t2}, $p$ cannot be in $S\setminus (T_1\cup T_2)$.
We conclude that $((T_1\setminus T_2)\cup (T_2\setminus T_1))\cap S=\emptyset$.
\end{clproof}

Suppose that $p_4\in V(G-S)$. 
We know that $p_2$ and $p_4$ are contained in some bags in $\btwn (B_1, B_2)$.
By symmetry, we assume $\abs{\btwn (B_1, A_1)}\le \abs{\btwn (B_1, A_2)}$, where $A_1$ and $A_2$ are bags containing $p_2$ and $p_4$, respectively.

Since $p_2p_4\notin E(G)$, $A_1$ is not a complete bag, and thus it is a star whose center is adjacent to $\comp (A_1, B_1)$.
In case when $P_1=P_2=B_2$, we may assume that $p_2$ is contained in the non-$S$-attached twin class. 
Then $p_1$ should be adjacent to $p_4$, contradiction.

We may assume that $p_2\in V(G-S)$ and $p_4\in (T_1\cap T_2)\cap S$, because $((T_1\setminus T_2)\cup (T_2\setminus T_1))\cap S=\emptyset$. Since $p_4\in T_1\cap T_2$, we have $p_2\notin C_2$.
We can observe that $p_2$ has no neighbors in $S$ as $p_2$ is contained in some bag in $\btwn (B_1, B_2)\setminus \{B\}$, and it is not contained in $C_1\cup C_2$.

If $p_1$ is adjacent to $c_2$, 
then $c_2$ is adjacent to the end vertices of an induced path $p_1p_2p_3p_4$, implying that $G$ has a small DH obstruction, which is a contradiction.
We may assume that $p_1$ is not adjacent to $c_2$.
One can observe that in this case, $p_1$ is in some bag of $\btwn (B_1, B_2)$, and thus
$p_1$ is adjacent to $p_3$. It is a contradiction.

\subparagraph{\textbf{Case 2.} $(T_1\cap T_2)\cap V(G-S)=\emptyset$ :}  
This implies that there are no complete bags in $\btwn (B_1, B_2)\setminus \{B_1, B_2\}$, and especially, if $B_1$ or $B_2$ is a complete bag, then it has no non-$S$-attached twin class.
We first claim that $B=B_1$, $B\neq B_2$ and $B$ is the parent bag of $B_2$.

\begin{claim}
$B=B_1$, $B\neq B_2$ and $B$ is the parent bag of $B_2$.
\end{claim}
\begin{clproof}
Suppose $B\neq B_1$, and let $B_1'$ be the parent bag of $B_1$.
As $C_1$ is complete to $C_2$, 
$B_1$ is a star whose center is adjacent to $B_1'$ or a complete bag. 
By Lemma~\ref{lem:smallbranch}, there is no child of $B_1$, and thus $B_1$ is a leaf bag.
Also, by Lemma~\ref{lem:smallbranch}, $B_1'$ has exactly two neighbor bags unless $B_1'=B$.

We observe that $B_1'$ should be a star whose center is adjacent to $B_1$.
When $B_1$ is a star, $B_1'$ is a star whose center is adjacent to $B_1$, as there is no complete bag in $\btwn (B_1, B_2)\setminus \{B_1, B_2\}$.
When $B_1$ is a complete bag, if $B_1'$ is a star whose leaf is adjacent to $B_1$, 
then we can apply Reduction Rule~\ref{rrule:leaf} because $B_1$ contains no non-$S$-attached twin class.
Thus $B_1'$ is a star whose center is adjacent to $B_1$.
Due to Lemma~\ref{lem:starcenter}, such a bag $B_1'$ cannot exist. 
Therefore, we have $B=B_1$. 

Since $B_1\neq B_2$ by Lemma~\ref{lem:twosattachedbags}, we have $B\neq B_2$. 
In the same reason, there are no bags in $\btwn (B, B_2)\setminus \{B, B_2\}$. It implies that $B$ is the parent of $B_2$.
\end{clproof}

  \begin{figure}[t]
    \centering
      \includegraphics[scale=0.5]{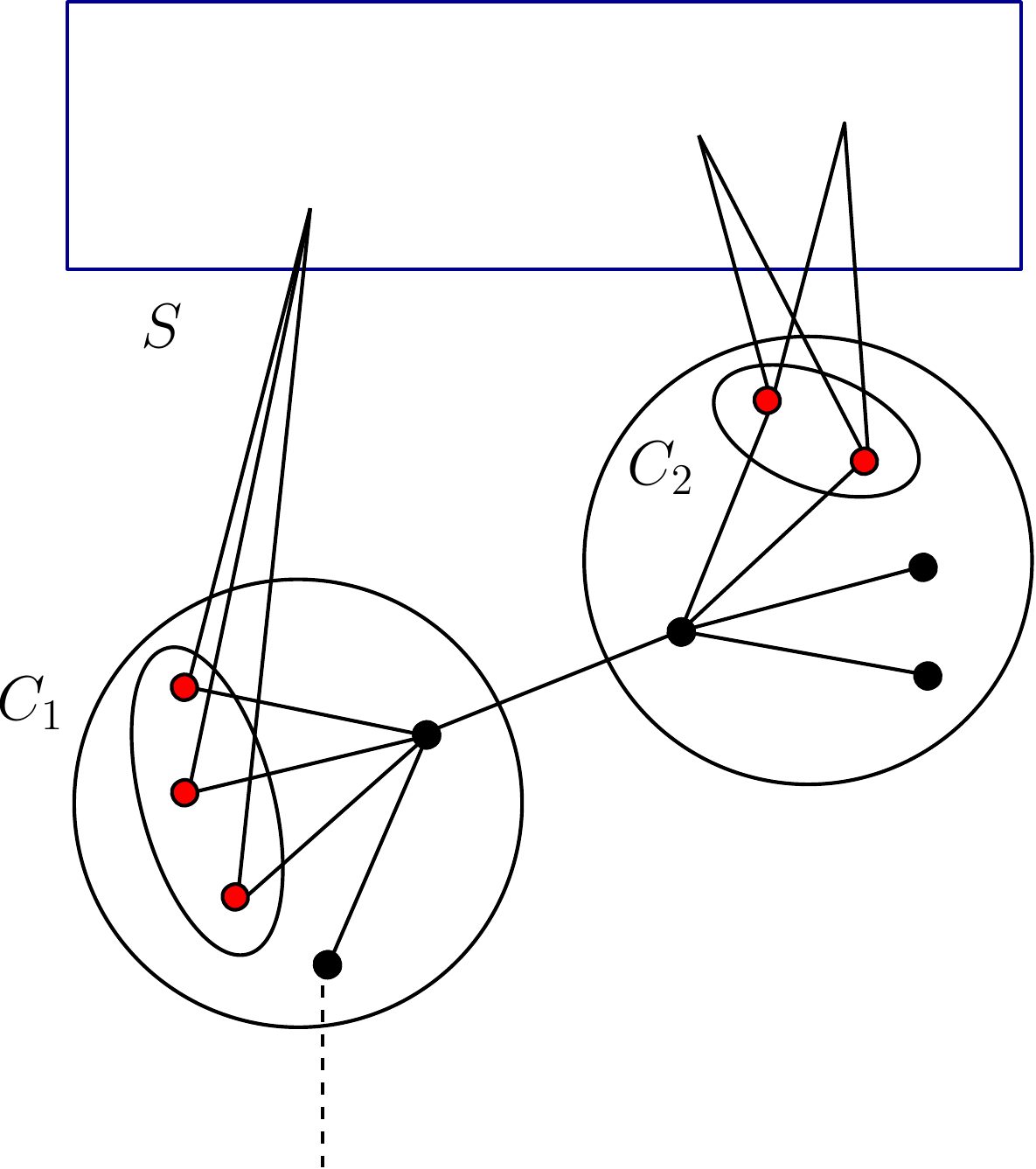}

    \caption{The case when $(T_1\cap T_2)\cap V(G-S)=\emptyset$ in Proposition~\ref{prop:complete2}.} \label{fig:prop422}
      
  \end{figure}

See Figure~\ref{fig:prop422} for an illustration.
Recall that $p_2$ is contained in $G-S$. Thus, $p_2$ should be contained in $C_2$ or
$B_2$ has the non-$S$-attached twin class and $p_2$ is contained in this class.

Suppose $p_2\in C_2$. Then $p_4\in (T_1\setminus T_2)\cap S$. 
If $p_4$ has a neighbor in $T_2$, then we have a bypassing vertex.
So, we may assume that $p_4$ has no neighbors in $T_2\cap S$.
As $C_1$ is complete to $C_2$ and Branching Rule~\ref{brule:reducecomponent} is exhaustively applied,
$p_4$ and $T_2\cap S$ are contained in the same connected component of $G[S]$.
Let $P$ be a shortest path from $p_4$ to $T_2\cap S$ in $G[S]$.
Then $Pp_2$ is an induced path of length at least $3$ and $p_3$ is adjacent to its end vertices, and therefore $G[S\cup \{p_2, p_3\}]$ contains a DH obstruction 
which contradicts the exhaustive application of Branching Rule~\ref{brule:threevertices}. 

Now, suppose $p_2\notin C_2$. It implies that $B_2$ is a star bag having a non-$S$-attached twin class, and $p_2$ is contained in the set.
Then $p_1$ should be a common neighbor of $c_2$ and $p_2$.  
Let $P$ be the shortest path from $p_4$ to $T_2\cap S$ in $G[S]$.
First assume that $p_1$ has a neighbor in $P$. Among neighbors of $p_1$ in $P$, we choose the vertex $q$ such that the distance between $p_4$ and $q$ in $P$ is shortest.
Let $Q$ be the subpath of $P$ from $p_4$ to $q$.
Then $p_2p_1Q$ is an induced path of length at least $3$ since $p_1$ is not adjacent to $p_4$, 
and $p_3$ is adjacent to its end vertices. 
Therefore, $G[S\cup \{p_1, p_2, p_3\}]$ contains a DH obstruction, contradicting our assumption that 
$G$ is reduced under Branching Rule~\ref{brule:threevertices}.
We may assume that $p_1$ has no neighbors in $P$.
In this case, 
$p_2p_1c_2P$ is an induced path of length at least $3$, and $p_3$ is adjacent to its end vertices.
Therefore, $G[S\cup \{p_1, p_2, p_3, c_2\}]$ contains a DH obstruction, contradicting our assumption that 
$G$ is reduced under Branching Rule~\ref{brule:threevertices}.

\medskip
We conclude that for every induced path $P=p_1p_2p_3p_4p_5$ with $p_3\in C_2$, 
  there is a bypassing vertex for $P$ and $p_3$.
\end{proof}

\begin{PROP}
\label{prop:complete1}
  Suppose $C_1$ is complete to $C_2$, $B\neq B_1$, and $B$ is a complete bag. 
  Then $B_1$ contains a non-$S$-attached twin class $C_1'$ and 
  for every induced path $P=p_1p_2p_3p_4p_5$ with $p_3\in C_1'$, 
  there is a bypassing vertex for $P$ and $p_3$. 
   \end{PROP}
  \begin{proof}
For each $i\in \{1,2\}$ and let $c_i\in C_i$ and $T_i=N_G(C_i)$.
Let $B_3$ be the parent bag of $B_1$.
As $C_1$ is complete to $C_2$,
$B_1$ is either a complete bag or a star whose center is adjacent to $B_3$.
We observe that $B_1$ has a non-$S$-attached class, and $(T_2\setminus T_1)\cap V(G-S)\neq \emptyset$.

\begin{claim}
$B_1$ has a non-$S$-attached class.
\end{claim}
\begin{clproof}
Suppose for contradiction that $B_1$ has no non-$S$-attached class, that is, its unmarked vertices form one $S$-attached twin class.
We verify that there is no child bag of $B_1$.
Suppose for contradiction that there is a child $B_1'$ of $B_1$ and let $D'$ be the component $\comp (B_1, B_1')$.
Since $D'$ contains no $S$-attached bag, by Lemma~\ref{lem:smallbranch}, $B_1$ should be a star whose center is adjacent to $B_1'$, which is a contradiction.
Also, every bag in $\btwn (B, B_1)\setminus \{B, B_1\}$  is not a $(B, B_1)$-separator bag, and by Lemma~\ref{lem:smallbranch},  it has exactly two neighbor bags.

Assume $B_1$ is a complete bag. Then its parent $B_3$ is a star 
and thus $B\neq B_3$. 
By Lemma~\ref{lem:starcenter},
 the center of $B_3$ is not adjacent to $B_1$.
Thus, the center of $B_3$ is adjacent to its parent.
As $B_1$ has no non-$S$-attached class, 
we can apply Reduction Rule~\ref{rrule:leaf} to reduce the split decomposition, which is a contradiction.
Assume $B_1$ is a star whose center is adjacent to its parent.
Similarly, by Lemma~\ref{lem:starcenter}, 
$B_3$ is not a star whose center is adjacent to $B_1$.
 Thus, we may assume the parent of $B_1$ is a complete bag, but in this case, we can apply Reduction Rule~\ref{rrule:leaf}.

We conclude that $B_1$ contains a non-$S$-attached twin class.
\end{clproof}

Let $C_1'$ be the non-$S$-attached twin class in $B_1$. As $C_1$ and $C_1'$ have the same neighborhood in $G-S$, 
$C_1'$ is complete to $C_2$.
  
\begin{claim}
$(T_2\setminus T_1)\cap V(G-S)$ contains a vertex that has no neighbors in $S$.
\end{claim}
\begin{clproof}
If $B_1$ is a star bag, then $C_1'\subseteq (T_2\setminus T_1)\cap V(G-S)$ and $y\in C_1'$.
If $B_1$ is a complete bag, then since $B$ is a complete bag, $B_3\neq B$ and the unmarked vertices in $B_3$ are contained in $(T_2\setminus T_1)\cap V(G-S)$.
Let $y$ be an unmarked vertex in $B_3$.
By the choice of $y$, $y$ has no neighbors in $S$.
\end{clproof}

Let 
$y$ be a vertex in $(T_2\setminus T_1)\cap V(G-S)$ having no neighbors in $S$.

Suppose there is an induced path $P=p_1p_2p_3p_4p_5$ with $p_3\in C_1'$. 
We will prove that there is a bypassing vertex for $P$ and $p_3$.
Let $D'$ be the connected component of $D-V(B)$ containing the parent of $B$.

We claim that $p_2$ and $p_4$ are contained in $(T_1\cap T_2)\cap V(G-S)$.
\begin{claim}
$p_2$ and $p_4$ are contained in $(T_1\cap T_2)\cap V(G-S)$.
\end{claim}
\begin{clproof}
Note that $p_2$ or $p_4$ is contained in either $D'$ or $\btwn (B_1, B_2)$.
If both $p_2$ and $p_4$ are contained in $D'$, then this is clear.
If both $p_2$ and $p_4$ are contained in $\btwn (B_1, B_2)$, then without loss of generality, 
we assume that $\abs{\btwn (B_1, A_1)}\le \abs{\btwn (B_1, A_2)}$ where $A_1$ and $A_2$ are bags containing $p_2$ and $p_4$, respectively.
Since there is no $(B_1, B_2)$-separator bag, $p_1$ should be adjacent to $p_4$, a contradiction.
Lastly, we assume that one of $p_2$ and $p_4$ is in $D'$, but the other is in $\btwn(B_1, B_2)$.
By symmetry we assume $p_2\in V(D')$. 
If $p_4\in \btwn (B_1, B)$, then $p_4$ is contained in a complete bag, and thus $p_2$ is adjacent to $p_4$.
If $p_4\in \btwn (B, B_2)$, then $p_4$ is clearly adjacent to $p_2$, as $B$ is a complete bag.
Both cases are not possible.

We conclude that $p_2$ and $p_4$ are contained in $(T_1\cap T_2)\cap V(G-S)$.
\end{clproof}

Suppose $(T_1\cap T_2)\cap S\neq \emptyset$. Let $x\in (T_1\cap T_2)\cap S$.
Since we know that $p_3$ has no neighbors in $S$, we have $p_3x\notin E(G)$, 
and by (1) of Lemma~\ref{lem:completerelationt1t2},
$x$ should be adjacent to both $p_2$ and $p_4$. Thus, $x$ is a bypassing vertex, as required.
We may assume that $(T_1\cap T_2)\cap S=\emptyset$.

Since $C_1$ is complete to $C_2$, by Branching Rule~\ref{brule:reducecomponent}, 
we know that $T_1\cap S$ and $T_2\cap S$ are contained in the same connected component of $G[S]$.
Let $P$ be a shortest path from $T_1\cap S$ to $T_2\cap S$.
If $P$ has length at least $2$, then $G[V(P)\cup \{c_1, c_2\}]$ is an induced cycle of length at least $5$, contradicting our assumption that $G$ is reduced under Branching Rule~\ref{brule:threevertices}.
Thus, $P$ has length $1$. Let $q_1q_2$ be the path where $q_i$ is a neighbor of $c_i$ for each $i\in \{1,2\}$.
 
 Note that $q_1\in T_1\setminus T_2$, $q_2\in T_2\setminus T_1$ and $p_2, p_4\in T_1\cap T_2$.
 If $q_1$ or $q_2$ is adjacent to one of $p_2$ and $p_4$, the by (1) of Lemma~\ref{lem:completerelationt1t2}, 
 it is adjacent to both $p_2$ and $p_4$.
 This means that it becomes a bypassing vertex, as required. 
 Therefore, we may assume that for each $i\in \{1,2\}$, $q_i$ is adjacent to neither $p_2$ nor $p_4$.
 So, $p_2c_1q_1q_2$ is an induced path of length $3$, and $c_2$ is adjacent to its end vertices. 
 It implies that $G$ has a small DH obstruction, which is a contradiction.

\medskip
We conclude that for every induced path $P=p_1p_2p_3p_4p_5$ with $p_3\in C_2$, 
  there is a bypassing vertex for $P$ and $p_3$.
\end{proof}

\begin{PROP}\label{prop:finish}
If $D$ is a reduced canonical split decomposition of a connected component of $G-S$,
then $D$ is empty.
\end{PROP}
\begin{proof}
Suppose a reduced canonical split decomposition $D$ of a connected component of $G-S$ contains a vertex.
If it contains at most one $S$-attached twin class, then Reduction Rule~\ref{rrule:dhcomponent} can be applied.
We may assume that $D$ contains at least two $S$-attached twin classes.

We choose a root bag, and let $B$ be a farthest bag from the root bag 
 such that there are two descendant bags $B_1$ and $B_2$ of $B$ having distinct $S$-attached twin classes $C_1$ and $C_2$, respectively.
By Lemma~\ref{lem:twosattachedbags},
$B_1\neq B_2$, and one of the following happens:
\begin{enumerate}[(1)]
\item The distance from $C_1$ to $C_2$ in $G-S$ is $2$ and the unique $(C_1, C_2)$-separator bag is contained in $\comp (B, B_i)$ for some $i\in \{1,2\}$.
\item  $C_1$ is complete to $C_2$ and either
\begin{itemize}
\item $B$ is a star bag and $C_i$ is the set consisting of the center of $B$ for some $i\in \{1,2\}$, or
\item $B$ is a star bag whose center is adjacent to $\comp (B, B_i)$ for some $i\in \{1,2\}$.
\end{itemize}
\item  $C_1$ is complete to $C_2$ and $B$ is a complete bag.
\end{enumerate}
If (1) happens, then by Proposition~\ref{prop:anticomplete2}, Reduction Rule~\ref{rrule:p5middle} is applied to remove $C_i$.
If (2) happens, then by Proposition~\ref{prop:complete2}, Reduction Rule~\ref{rrule:p5middle} is applied to remove $C_{3-i}$.
If (3) happens, then by Proposition~\ref{prop:complete1},  Reduction Rule~\ref{rrule:p5middle} is applied to remove the non-$S$-attached twin class in one of $B_1$ and $B_2$.
But this contradicts the assumption that $D$ is reduced.
\end{proof}

\section{The Algorithm, Lower Bounds and Applications}
\label{sec:completing}

Our goal in this section is to give a proof of our main result, Theorem~\ref{thm:main}, and obtain corresponding lower bounds. 

\subsection{The Algorithm}
Below, we use the presented reduction and branching rules to give an algorithm for \disjointDHVD. This is then followed by a proof of our main algorithmic result.

\begin{theorem}
\label{thm:main2}
\disjointDHVD\ can be solved in time $\bigoh(6^{k+\cc(G[S])}\cdot |V(G)|^{6}(|V(G)|+E(G)))$.
\end{theorem}
\begin{proof}
Let $(G, S, k)$ be an instance of \disjointDHVD. 
We exhaustively apply Branching Rules~\ref{brule:threevertices}--\ref{brule:reducecomponent} and Reduction Rules~\ref{rrule:dhcomponent}--\ref{rrule:bypassing2}.
We prove that one of rules can be applied until $G-S$ is empty or $k$ becomes $0$.
In both cases, we can test whether the resulting instance is distance-hereditary or not in polynomial time, and output an answer.
Suppose $k$ does not reach $0$. Then by Proposition~\ref{prop:finish}, $G-S$ contains no vertices.
Therefore, the resulting instance satisfies that $G-S$ is empty, as mentioned.

We argue that the runtime bounds hold. For convenience, we will denote $|V(G)|$ by $n$ and $|E(G)|$ by $m$.
First notice that each branching rule reduces either $k$ or the number of connected components in $G[S]$ and branches into at most $6$ subinstances. Moreover, none of the reduction rules change $k$ or the number of components in $S$. Hence a branching rule is applied at most $k+\cc(G[S])$ times. Similarly, every reduction rule reduces either the number of vertices in $G-S$ or the number of bags in canonical split decomposition of $G-S$. Therefore, it is not hard to observe that the branching tree of the algorithm will have at most $6^{k+\cc(G[S])}$ leaves and each leaf will be in depth at most $\bigoh(n)$ and hence the branching tree will have at most $\bigoh(6^{k+\cc(G[S])}\cdot n)$ nodes. In the following we will discuss that the runtime in every node will not exceed $\bigoh(n^{5}(n+m))$. In each node, we go through the branching and reduction rules, in the order they are introduced in the paper, and apply the first rule that can be applied. 
Let us start with detecting and applying  Branching Rule~\ref{brule:threevertices}.  
Our algorithm is going through all sets $X\subseteq G-S$ of size at most $5$ and checking, whether $G[S\cup X]$ is distance-hereditary. It follows from Theorems~\ref{thm:dahlhaus} and \ref{thm:bouchet} that we check whether a graph is distance-hereditary in time $\bigoh(n+m)$. If the graph is not distance-hereditary, application of the rule can be done in constant time. Hence, the Branching Rule~\ref{brule:threevertices} can be verified in time  $\bigoh(n^{5}(n+m))$. Similarly, for Branching Rule~\ref{brule:reducecomponent} for every set $X\subseteq V(G-S)$ of size at most $5$, we can in time $\bigoh(n+m)$, e.g. using breadth-first search, verify that the neighborhood of $X$ is in the same connected component and the same running bound follows.  
After verifying that the graph is actually reduced under Branching Rules~\ref{brule:threevertices} and \ref{brule:reducecomponent} it follows from Proposition~\ref{prop:applicationRR} that we can in time $\bigoh(n^{6})$ either apply one of Reduction Rules~\ref{rrule:dhcomponent}--\ref{rrule:bypassing2} or correctly deduce that the graph is reduced also under Reduction Rules~\ref{rrule:dhcomponent}--\ref{rrule:bypassing2}. Hence, the whole algorithm for \disjointDHVD\ can be implemented in time $\bigoh(6^{k+\cc(G[S])}\cdot n^{6}(n+m))$.
\end{proof}

\begin{theorem}
\label{thm:main}
\DHVD\ can be solved in time $\bigoh(37^k\cdot |V(G)|^7(|V(G)|+|E(G)|))$.
\end{theorem}
\begin{proof}
We apply the standard iterative compressing technique. The algorithm involves a two-step reduction of \DHVD: we first reduce \DHVD\ to the \textsc{Compression} problem, which reduces to
\disjointDHVD. 

For convenience, we denote for this proof $|V(G)|=n$ and $|E(G)|=m$.
Fix an arbitrary labeling $v_1,\dots, v_n$ of $V(G)$ and let $G_i$ be a the graph $G[\{v_1,\dots, v_i\}]$ for $1\le i\le n$. From $i=1$ up to $n$, we consider the following the \textsc{Compression} problem for \DHVD: given a graph $G_i$ and $S_i\subseteq V(G_i)$ such that $G_i-S_i$ is distance-hereditary and $|S_i|\le k+1$, we aim to find a set $S_i'\subseteq V(G_i)$ such that $G_i-S'_i$ is distance-hereditary and $|S'_i|\le k$, if one exists, and output \textsc{No} otherwise. Since distance-hereditary graphs are closed under taking induced subgraphs, $(G,k)$ is \textsc{Yes}-instance of \DHVD\ if and only if $(G_i,S_i)$ is a \textsc{Yes}-instance for 
\textsc{Compression} for all $i$, where $(G_i,S_i)$ is a legitimate instance. Hence we correctly output that $(G,k)$ is \textsc{No}-instance of \DHVD\ if $(G_i,S_i)$ is a  \textsc{No}-instance for some $i$. Moreover, if $S_i'$ is a solution to the $i$-th instance of  \textsc{Compression}, then $(G_{i+1}, S_i'\cup\{v_{i+1}\})$ is a legitimate instance for $(i+1)$-th
instance of  \textsc{Compression}. 

Given an instance $(G, S)$ of \textsc{Compression}, we enumerate all possible intersections $I$ of $S$ and a desired solution to $(G,S)$. For each guessed set $I$, we solve the instance $(G-I, S\setminus I, k-|I|)$ of \disjointDHVD\ using Theorem~\ref{thm:main2}. Note that  $(G, S)$ is \textsc{Yes}-instance if and only if  $(G-I, S\setminus I, k-|I|)$ is \textsc{Yes}-instance for some $I\subseteq S$. If $S'$ is a solution to  $(G-I, S\setminus I, k-|I|)$, then clearly $S'\cup I$ is a solution to the instance $(G,S)$ of \textsc{Compression}. Conversely, if $S'$ is a solution to  the instance $(G,S)$ of \textsc{Compression} then for the set $I=S\cap S'$ the instance $(G-I, S\setminus I, k-|I|)$ is  \textsc{Yes}-instance for \disjointDHVD. Therefore, using the algorithm from Theorem~\ref{thm:main2} we can correctly solve \DHVD.

It remains to prove the complexity of the algorithm. Given an instance $(G,S)$ we guess at most ${{k+1}\choose {i}}$ sets $I\subseteq S$ of size $i$ for each $1\le i\le k$. Note that $S\setminus I$ has size at most $k+1-i$, and in particular $G[S]$ has at most $k+1-i$ connected components. Therefore, we can solve the resulting instance $(G-I, S\setminus I, k-i)$ of \disjointDHVD\ in time $\bigoh(6^{2k-2i+1}\cdot n^{6}(n+m))=\bigoh(36^{k-i}\cdot n^{6}(n+m))$. Summing up, \DHVD\ can be solved by running an algorithm for \textsc{Compression} at most $n$ times, which yields the claimed running time $$n\cdot \sum_{i=0}^k{{k+1}\choose {i}} \cdot \bigoh(36^{k-i}\cdot n^{6}(n+m)) = \bigoh(37^{k}\cdot n^{7}(n+m)).$$

Note that the equality follows from the use of the binomial theorem, which states that $\sum_{i=0}^{n}\binom{n}{i}a^ib^{n-i}=(a+b)^n$ (see, e.g., Chapter 10 in Cygan~et~al.~\cite{CyganFKLMPPS15}).
\end{proof}

\subsection{Lower Bounds}
Here we will present our lower bound result, based on the well-established exponential time hypothesis~\cite{ImpagliazzoRF2001}.

\begin{definition}[Exponential Time Hypothesis (ETH)]
 There exists a constant $s>0$ such that \textsc{3-CNF-SAT} with $n$ variables and $m$ clauses cannot be solved in time $2^{sn}(n+m)^{\mathcal{O}(1)}$.
\end{definition}

Our result uses the fact that the classical \textsc{Vertex Cover} problem cannot be solved in subexponential time under ETH.
  
\begin{theorem}[Cai and Juedes \cite{CaiJuedes03}]\label{thm:ethVC}
There is no $2^{o(k)}\cdot |V(G)|^{\mathcal{O}(1)}$ algorithm for \textsc{Vertex Cover}, unless ETH fails.
\end{theorem}

\begin{theorem}
\label{thm:main3}
There is no $2^{o(k)}\cdot |V(G)|^{\mathcal{O}(1)}$ algorithm for \DHVD\, unless ETH fails.
\end{theorem}
\begin{proof}
For a graph $G$, we will denote $|V(G)|$ by $n$ and $|E(G)|$ by $m$.
For contradiction suppose there exists an algorithm for solving the \DHVD\ problem in time $2^{o(k)}\cdot n^{\mathcal{O}(1)}$. We show that we can solve \textsc{Vertex Cover} in time $2^{o(k)}\cdot n^{\mathcal{O}(1)}$. Let $(G, k)$ be an instance of \textsc{Vertex Cover} problem. We construct a graph $G'$ as follows. We replace every edge $uv$ of $G$ with two vertex disjoint paths of length $3$ between $u$ and $v$. Note that for every edge $uv$ in $G$ the two disjoint paths of length $3$ in $G'$ form an induced subgraph isomorphic to $C_6$. Moreover we have $|V(G')| = |V(G)|+4|E(G)|$. We claim that $G$ has a vertex set $S$ of size at most $k$ such that $G-S$ has no edges if and only if $G'$ has a vertex deletion set of size at most $k$ to a distance-hereditary graph. Suppose that $G$ has such a vertex cover $S$. It is easy to confirm that $G'-S$ is a  disjoint union of subdivisions of stars, which is distance hereditary. 

For the converse direction, suppose $G'$ has a distance-hereditary vertex deletion set $S$ of size at most $k$. 
Let us fix an arbitrary edge $uv$ in $G$. Note that no DH obstruction contains a pendant vertex. Hence we observe that if $H$ is a DH obstruction containing a vertex $t$ on a shortest $u-v$ path in $G'$, then $H$ contains both vertices $u$ and $v$ as well. Therefore, if $t\in S$, then also graphs $G'-(S\setminus\{t\}\cup \{u\})$ and $G'-(S\setminus\{t\}\cup \{v\})$ are distance-hereditary. Since the choice of the edge $uv$ was arbitrary, we can find a set $T$, such that $T\subseteq V(G)$, $|T|\le |S|$, and $G'-T$ is a distance-hereditary graph. Clearly for every edge $uv$ in $G$, $T$ contains $u$ or $v$, otherwise $G'-T$ contains an induced $C_6$. We conclude that $T$ is a vertex cover of $G$, which finishes the proof. 
\end{proof}

\subsection{Example Applications}
There is an established line of research studying the algorithmic applications of vertex deletion sets to specific graph classes~\cite{GajarskyHOORRVS13,EibenGanianSzeider15,EibenGanianSzeider15b,FellowsLMRS08}.
In this context, it is natural to ask whether Theorem~\ref{thm:main} allows the development of single-exponential algorithms for problems parameterized by the size of a vertex deletion set to distance-hereditary graphs. 

Clearly, any problem that is FPT when parameterized by clique-width (and rank-width) must also be FPT when parameterized by the size of a vertex deletion set to distance-hereditary graphs. However, the existence of a single-exponential FPT algorithm parameterized by clique-width does not immediately imply that the problem also admits a single-exponential FPT algorithm parameterized by our parameter, since the addition of $k$ vertices to a graph may increase clique-width by a factor of up to $2^k$~\cite{Gurski2016}. On the other hand, known FPT algorithms parameterized by rank-width usually do not have a single-exponential dependency on the parameter. As a consequence, one cannot obtain the following examples of single-exponential algorithms by simply solving these problems via known FPT algorithms parameterized by rank-width or clique-width.

\begin{lemma}
 \textsc{Vertex Cover} %
and \textsc{$3$-Coloring} admit a single-exponential FPT algorithm when parameteried by the size of a vertex deletion set to distance-hereditary graphs.
\end{lemma}

\begin{proof}
For each of the presented problems, we always begin by invoking Theorem~\ref{thm:main} to compute a vertex deletion set $X$ to distance-hereditary graphs of size at most $k$.

In the case of \textsc{Vertex Cover}, we can apply standard branching algorithms to solve the problem. In particular, we begin by branching over the at most $2^k$ options of how $X$ intersects with a (hypothetical) solution; let $X_1$ be one such subset of $X$ and let $X_2=X\setminus X_1$. 
After branching we proceed by testing the validity of a branch (i.e., whether each edge with both endpoints in $X$ is covered by $X_1$). For each valid branch, we delete $X$ and the set $Z$ of all neighbors of $X_2$ in $G-X$. Next, we find a minimum vertex cover $C$ in the remaining distance-hereditary subgraph of $G$ in polynomial time. Finally, for each branch we compare the desired solution size with $|C\cup X_1\cup Z|$; clearly, a  graph is a YES-instance of \textsc{Vertex Cover} if and only if at least one selection of $X_1$ results in a value of $|C\cup X_1\cup Z|$ which is at most the desired solution size.

For \textsc{$3$-Coloring}, we also begin by branching over the at most $3^k$ $3$-colorings of $X$. For each such proper $3$-coloring of $X$, we construct an instance of \textsc{$3$-List Coloring} as follows: the input graph is $G-X$, and the list of admissible colors for each vertex $v$ contains all colors that are not used by a neighbor of $v$ in $X$. The \textsc{$3$-List Coloring} problem can be solved in polynomial time on distance-hereditary graphs: indeed, the problem can easily be reduced to the MSO$_1$ model checking problem over labeled graphs with (at most) $8$ labels. Since $G-X$ has rank-width at most $1$, the polynomial-time tractability of the problem follows for instance from Courcelle's Theorem~\cite{CourcelleMR00}. All that remains now is to test whether at least one of the considered $3^k$ branches give rise to a yes-instance of \textsc{$3$-List Coloring} on $G-X$.
\end{proof}

\section{Concluding Notes}
We conclude with a few remarks on why we believe that the presented algorithm is of high interest. First, it intrinsically exploits the properties guaranteed by distinct, seemingly unrelated characterizations of distance-hereditary graphs; this approach can likely be used to design or improve algorithms for other vertex deletion problems. Second, it uses highly nontrivial reduction rules which simplify canonical split decompositions, and an adaptation or extension of the presented rules could be highly relevant for other graph classes characterized by special canonical split decompositions, such as parity graphs~\cite{CiceroneS1999} or circle graphs~\cite{GSH1989}. Third, it is the first of its kind which targets a ``full'' class of graphs of bounded rank-width (contrasting previous results for specific subclasses of graphs of rank-width~$1$~\cite{HuffnerKMN2010, AgrawalKLS16,KimK2015, KanteKKP2015}).

It is worth noting that there remains a number of interesting open problems in this general area. Perhaps the most prominent one is the question of whether vertex deletion to graphs of rank-width $c$, for any constant $c$, admits a single-exponential fixed-parameter algorithm. Our algorithm represents the first steps in this general direction. 
Recently, Kim and the third author~\cite{KimK2016} announced a polynomial kernel for \textsc{Distance-Hereditary Vertex Deletion}.
The existence of a polynomial kernel or an approximation algorithm for such vertex deletion problems for $c>1$ remains open.

\section*{Acknowledgment}
The first and second authors were supported by the Austrian Science Fund (FWF, projects P26696 and W1255-N23). Robert Ganian is also affiliated with FI MU, Brno, Czech Republic. The third author was supported by ERC Starting Grant PARAMTIGHT (No. 280152) and also supported by the European Research Council (ERC) under the European Union's Horizon 2020 research and innovation programme (ERC consolidator grant DISTRUCT, agreement No. 648527). 

The authors would like to thank the anonymous reviewers for helpful suggestions. In particular, one reviewer suggested merging several reduction rules to one, and this makes it easier to understand the algorithm.
The third author would like to thank Eun Jung Kim and Sang-il Oum for initial discussions on this problem.

\end{document}